\RequirePackage{fix-cm}
\documentclass[smallextended]{svjour3}       
\smartqed  
\usepackage{graphicx}
\usepackage{amsmath}
\usepackage{amssymb}
\journalname{}

\begin{document}

\title{Unchecked strategy diversification and collapse in continuous voluntary public good games 
}

\titlerunning{Continuous voluntary public good games}        

\author{Tatsuya Sasaki         
\and
         {\AA}ke Br{\"a}nnstr{\"o}m
\and
         Isamu Okada
\and
         Tatsuo Unemi
}


\institute{T. Sasaki \at
	Facluty of Mathematics, University of Vienna, 1090 Vienna, Austria \\
	Evolution and Ecology Program, International Institute for Applied Systems Analysis, 2361 Laxenburg, Austria \\
	Graduate School of Engineering, Soka University, 192-8577 Tokyo, Japan \\
           Tel.: +43-1-4277-50774\\
           \email{tatsuya.sasak@univie.ac.at}
        \and
	{\AA}. Br{\"a}nnstr{\"o}m \at
           Department of Mathematics and Mathematical Statistics, University of Ume{\aa}, 90187 Ume{\aa}, Sweden \\     
           Evolution and Ecology Program, International Institute for Applied Systems Analysis, 2361 Laxenburg, Austria
        \and
	I. Okada \at
           Department of Business Administration, Soka University, 192-8577 Tokyo, Japan \\
        \and
	T. Unemi \at
           Department of Information Systems Science, Soka University, 192-8577 Tokyo, Japan \\
}

\date{Received: date / Accepted: date}

\maketitle

\begin{abstract}
Cooperation or defection and participation or withdrawal are well-known options of behavior in game-like activities in free societies, yet the co-evolutionary dynamics of these behavioral traits in the individual level are not well understood. 
Here we investigate the continuous voluntary public good game, in which individuals have two types of continuous-valued options: a probability of joining the public good game and a level of cooperative investment in the game. 
Our numerical results reveal hitherto unreported phenomena: 
(i) The evolutionary dynamics are initially characterized by oscillations in individual cooperation and participation levels, in contrast to the population-level oscillations that have previously been reported. 
(ii) Eventually, the population's average cooperation and participation levels converge to and stabilize at a center. 
(iii) Then, a most peculiar phenomenon unfolds: The strategies present in the population diversify and give rise to a ``cloud'' of tinkering individuals who each tries out a different strategy, and this process continues unchecked as long as the population's cooperation and participation levels remain balanced.
Over time, however, imbalances build up as a consequence of random drift and there is a sudden and abrupt collapse of the strategy-diversity cloud. The process then repeats again in a cyclic manner. 
To understand the three aforementioned phenomena, we investigate the system analytically using adaptive-dynamics techniques. Our analysis casts light on the mechanisms which underpin the unexpected and surprising evolutionary dynamics.

\keywords{social dilemmas \and evolution of cooperation \and voluntary participation \and adaptive dynamics \and diversification}
\subclass{91A06 \and 91A22 \and 91A40}
\end{abstract}

\section{Introduction}

The reasons for the emergence and maintenance of cooperative behavior is an enduring puzzle in biology and the social sciences \cite{Tri71,WiSo94,Har68,Ost90}. The nature of this puzzle is often characterized as follows: groups of cooperators outperform groups of defectors, but defectors always outperform cooperators in any mixed group \cite{Dawes}. This represents a social-dilemma situation: individual interests and the communal benefit are incompatible. Many theoretical and experimental investigations of cooperative behavior have employed the framework of the public goods game \cite{Bin94,Cha11}. Typical public good games can be described as follows: cooperators in an $N$-player group with $N \ge 2$ invest the same amount $c>0$ to a public pool and defectors do nothing at all; then, the total amount of investment is multiplied by a factor $r>1$ and distributed equally among all members of the group despite the preceding different investment levels \cite{Sig10}. With $r<N$, it follows that the marginal profit $c(r/N-1)$ is negative and thus each player is better off defecting than cooperating, irrespective of the other players’ decisions, i.e., defection dominates cooperation. Cooperators are thus evolutionarily doomed by defectors.
This is in contrast to the case $r>N$ in which each player is better off cooperating and, consequently, no social dilemma exists. 

Here we focus on the effects of voluntary participation (or exit) on the evolution of cooperation \cite{Hir70,OrbDaw93,BatKit95,Hayashi}. In modern societies individuals often have a great deal of freedom and anonymity. This allows an individual to get away with not only free-riding, but also opting-out types of behavior. Voluntary participation appears to be the simplest mechanism for promoting cooperation that can be justified as an \textit{a priori} option even under complete anonymity. This mechanism has been studied in the standard voluntary public good game consists of three pure strategies: cooperation, defection, and nonparticipation. 

The latter strategy is used by players who opt out of public good games and instead constantly earn a payoff $\sigma>0$ \cite{HauEtal02a,HauEtal02b}. If successful strategies are assumed to increase in frequency, for example through imitation dynamics, a ``rock-scissors-paper''-type rotational change in dominating strategy arises among the three pure strategies.  Cooperation can therefore be maintained through population-level oscillations in the relative frequencies of these three pure strategies.

The standard voluntary public goods game assumes that each player has one of three pure behavioral strategies \cite{HauEtal02a,HauEtal02b}. In many situations, it is more plausible to assume that players differ continuously in cooperative investments and participation probability, but the consequences of this assumption has thus far not been explored. Here, we investigate a continuous voluntary public goods game in which individuals are able to make continuously varying degrees of investment levels \cite{wahl1,wahl2,killingback3,killingback1,sherratt,Doe04,Bro08,Ake,CreEtal11,Deng11,Zhang,Par13} and also alter their participation rates \cite{SasEtal07,Chen08}. Surprisingly, we find that the emerging cooperative dynamics are very different from those of the standard public good game and in particular involve a hitherto unreported phase of strategy expansion and collapse. By drawing on both evolutionary game theory and the theory of adaptive dynamics, we can analytically understand and explain almost all of the unfolding phenomena, thus revealing important insights for the evolution of cooperation. 

This paper is organized as follows. 
In Sect. 2, we describe a continuous voluntary public good game and provide an individual-based model of the gradual evolution of cooperation and participation in public good games.
In Sect. 3, we conduct numerical simulations of the individual-based model which demonstrate cyclic oscillation, convergence to  the center, growth of ``cloud'' and its collapse.
In Sect. 4, we then compare those results with the theoretical predictions. 
We determine the expected payoff and invasion fitness and analyze the selection gradient and equilibria for monomorphic populations. To explore effects of small yet finite mutations, we also consider a geometrical analysis and polymorphic populations.
Finally, in Sect. 5, we provide further discussion.   

\section{Model description}

\setcounter{equation}{0}
\setcounter{figure}{0}

\subsection{Continuous voluntary public good game}

We consider a well-mixed population. An individual has a continuous strategy involving two traits $(c,p)$ in $U:=[0,1]^2$. The first coordinate $c$ represents the amount of investment that the individual makes in the public good game. The second coordinate $p$ represents the probability of participation in this game. 
For each game $N$ individuals with $N \ge 2$ are randomly selected from the population. Each of the $N$-players first determines whether to participate in the public good game or not, with one's own probability $p$. 
Those who participate can contribute an investment at a cost $c$ to themselves. All individual contributions are added up and multiplied with a factor $r$ with $r>1$. This amount is then shared equally among all participants. Each participant's payoff is given as a net benefit that consists of his or her share less the amount invested. 
Individuals who do not participate in the public good game instead receive a small payoff $\sigma$ with $0<\sigma<r-1$ that is independent of outcomes of the public good game. We require $\sigma > 0$ to ensure that nonparticipation is better off than a group of those who make no investment ($c=0$), and that $\sigma<r-1$, that a group of those who make full investment ($c=1$) is better off than nonparticipation. 
We assume that if there is only one participant, this single player has to act as a nonparticipant and therefore the payoff is $\sigma$.

It is well known that in the case of compulsory participation ($p=1$ for all individuals), no contribution with $c=0$ is only Nash equilibrium for $r < N$ \cite{CreEtal11}.
A focal participant $i_0$ with $c=c_0$ will earn in a game with $N-1$ co-players with cooperation levels $\{c_1, \cdots , c_{N-1}\}$ the payoff
$P(i_0) = \frac{r}{N} \sum_{k=0}^{N-1} c_k - c_0.$ 
Thus, the marginal profit per unit increase in the focal player's contribution is given by $dP/dc_0 = r/N-1$. With $r < N$, this is negative and each participant is tempted to reduce own contribution to zero: $c=0$ is (weakly) dominant, or with $r>N$, the marginal profit is positive and $c=1$ is (weakly) dominant. Therefore, we hereafter concentrate on the most stringent case with $r<N$.   

\subsection{Individual-based model}

In the individual-based model, strategies spread in a finitely large population with size $M$ by imitation and exploration (``social learning''). 
We assume that individuals are more likely to imitate strategies of those who have earned higher payoffs. 
For simplicity, we straightforwardly apply the replicator dynamics to the finite population as in Doebeli et al. \cite{Doe04}.

We consider asynchronous sequential updating of the finite population, as follows.
First, a focal individual $i_0$ is randomly picked up from the population. 
The $i_0$'s payoff $P(i_0)$, then, is determined through an interaction with $N-1$ co-players selected randomly, $\{i_1, \cdots , i_{N-1}\}$. 
After making a participating-decision with one's own participation rate,  
the focal individual's payoff is settled as
\begin{equation}
\label{eq:A.38}
P(i_0)=
\begin{cases}
\displaystyle
\frac{r}{S} \sum_{k=0}^{S-1} c_k - c_0 & \text{if $i_0$ participates and has a co-player,} \\
\sigma & \text{otherwise,}
\end{cases}
\end{equation}
where $c_k$ denotes the $i_k$'s investment level and the first to ($S-1$)-th players among the $N-1$ co-players ($1< S \le N$) are participants.

For comparison, another model individual $j_0$ is randomly chosen, and then, its payoff $P(j_0)$ is determined as in Eq. (\ref{eq:A.38}) through an interaction with random $N-1$ co-players selected independent of the $i_0$'s case.
Whether the focal individual $i_0$ imitates the model $j_0$ is determined with a probability $w$ that is proportional to those payoff difference:
$w = \frac{P(j_0) - P(i_0)}{\alpha}$ if $P(j_0) > P(i_0)$; otherwise, if $P(j_0) \le P(i_0)$, $w=0$.
In the former case, $\alpha$ is fixed as $\frac{r(N-1)}{N} - (\frac{r}{N}-1)$, which denotes the maximum degree among available payoff differences, and thus, ensures that $w\leq1$. Its first term represents the payoff of a full-defector (with $c=0$) within $N-1$ full-cooperators (with $c=1$), and the second term represents the payoff of a full-cooperator within $N-1$ full-defectors.

Finally, after the imitation event, the exploration can happen with a small probability $\mu$. In the case each trait is replaced by a value drawn independently from a Gaussian distribution with the former value of the trait as mean and a small standard deviation $s$.

\begin{figure}
\begin{center}
\includegraphics[scale=0.80]{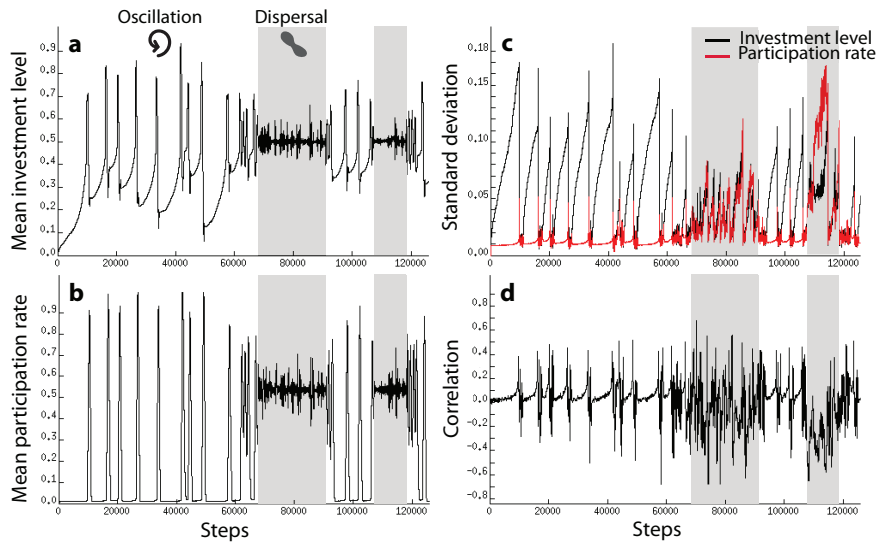}
\end{center}
\caption{ \label{fig:3}
Evolutionary history consisting of oscillation and dispersal.  
The initial state is monomorphic at $(c,p)=(0,0)$.
Mean investment level $\bar{c}$ (\textbf{a}), mean participation rate $\bar{p}$ (\textbf{b}),
standard deviations of investment levels (\textit{black}) and participation rates (\textit{red}) (\textbf{c}),
and 100-step moving average of the correlation between the two traits $c$ and $p$ (\textbf{d}) are provided.
Coevolution of cooperation and participation is characterized by transitions between the oscillation (\textit{white} intervals) and dispersal (\textit{gray} intervals) phases. 
The mean and the standard deviation of investment levels first increases gradually because of neutral drift along the boundary $p=0$, and the mean participation rate $\bar{p}$ remains close to 0, and then, increases rapidly when the mean investment $\bar{c}$ reaches about 0.5. After cycles for about $0.7\!\times\!10^5$ steps, the population attains a small neighborhood of the center Q.
Interestingly, from $0.7\!\times\!10^5$ to $0.9\!\times\!10^5$ steps, another dynamical phase ensues, in which $(\bar{c},\bar{p})$ is close to Q and slightly oscillates with quite a small amplitude.
Insufficient selection pressure there allows the trait distribution to spread and increase in variance. 
Another dispersal phase appears between $1.1\!\times\!10^5$ and $1.2\!\times\!10^5$ steps after the second oscillation phase.
From \textbf{c} and \textbf{d} one can observe that diversification can subsequently continue when the trait distribution keeps its correlation sufficiently negative. 
In particular, the second dispersal phase shows that a sudden contraction of diversity could occur when the negative correlation is lost.
Parameter: $M =10^4$, $\mu=0.001$, $s=0.005$, $N=5$, $r=3$, and $\sigma=1$.         
}
\end{figure} 

\begin{figure}
\begin{center}
\includegraphics[scale=0.80]{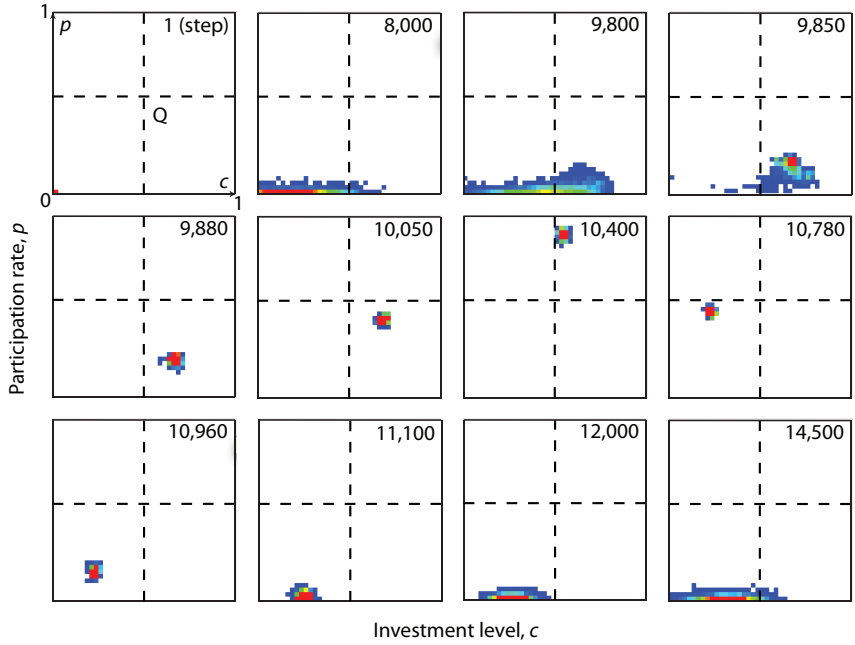}
\end{center}
\caption{\label{fig:4}
Snapshots of the oscillation phase corresponding to Fig. \ref{fig:3} (up to 14,500 steps). 
Each panel shows a snapshot of the frequency distribution of strategies (from high to low: \textit{red}, \textit{orange}, \textit{yellow}, \textit{green}, \textit{blue}, \textit{white} (for 0)) in the $(c,p)$ space at a certain step. 
The interior equilibrium Q is located at the intersection of the \textit{dashed} nullclines $c=0.5$ and $p \approx 0.5387$, beyond which the selection pressure on $p$ and $c$, respectively, changes.
The population is initially monomorphic at $(c,p)=(0,0)$.
Neutral drift first drives diversification of the investment level $c$ along the boundary $p=0$. 
As the distribution reaches around the point $(0.5,0)$, mutants with $c>0.5$ and $p>0$ can happen and then successfully invade.
Such mutants eventually displace the residents, and monomorphism is re-established. From there the population synchronously orbit the center Q. 
The orbit present is so distant from Q that the boundary $p=0$ absorbs it at the end. 
Then, neutral drift along $p=0$ begins, again.         
}
\end{figure} 

\begin{figure}
\begin{center}
\includegraphics[scale=0.80]{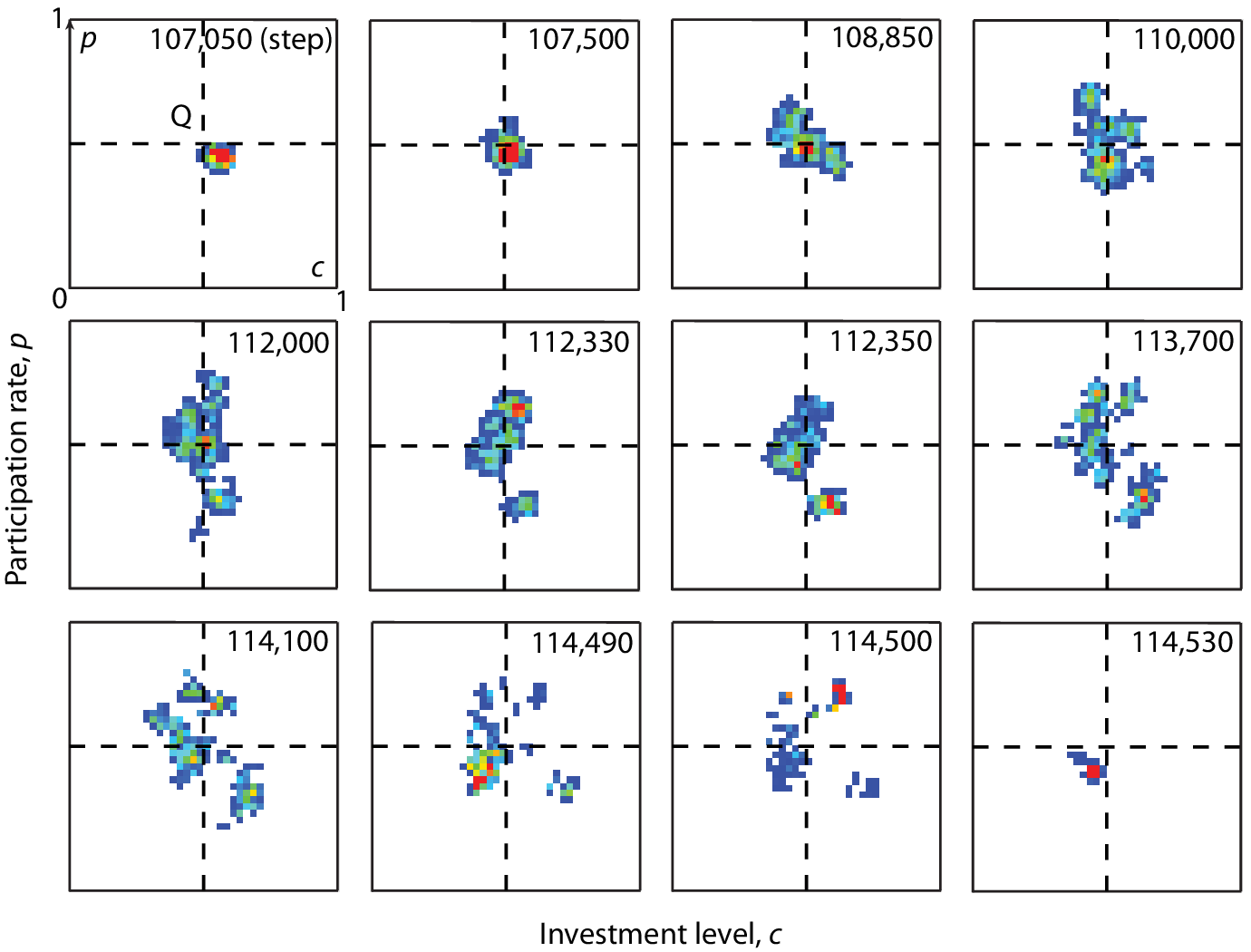}
\end{center}
\caption{\label{fig:5}
Snapshots of the dispersal phase corresponding to Fig. \ref{fig:3} (from 107,050 to 114,530 steps).
Panel features are the same as for Fig. \ref{fig:4}.
The cluster is already close to the center Q at 107,050 steps. After converging to Q, the gradual radiation derived from neutral explorations starts across Q while its trend remains negative. 
This dispersal process leads to splitting the residents into some distant subclusters after 112,000 steps.
Finally, the lower-right cluster that plays a major part in keeping the trend negative, vanishes accidentally. This subsequently causes a dramatically rapid contraction of diversity, and then, the population returns to monomorphism (114,530 steps).         
}
\end{figure}

\section{Results}

The individual-based model can embody surprisingly rich evolutionary processes. These are characterized by transitions between two qualitatively different dynamic phases, cyclic oscillation and diagonal dispersal. 
Figure \ref{fig:3} presents a typical sequence of the oscillation and dispersal phases and those corresponding snapshots are displayed in Figs. \ref{fig:4} and \ref{fig:5}, respectively. In particular, Fig. \ref{fig:6} demonstrates a trial in which investment levels largely diverge into full cooperation ($c=1$) and defection ($c=0$) (see also corresponding snapshots in Fig. \ref{fig:7}). We now explain each of the two dynamic phases in turn. 

\begin{description}
\item[(i) \textit{Cyclic oscillation.}]

Let the population start with a monomorphic state in which all individuals adopt the same state as $(c,p)=(0,0)$ (Figs. \ref{fig:3} and \ref{fig:4}). In the situation, any mutant with respect to the investment level $c$ is able to invade the resident through neutral drift, leading to a gradual increase in the deviation of $c$, while the participation rates $p$ are kept to a low level corresponding to the probability $\mu$ and standard deviation $s$ of explorations. Then, about when the strategy distribution reaches the critical investment level $c_\textrm{Q}$ given by $\sigma/(r-1)$ in Sect. 4.1, beyond which the sign of the selection pressure on $p$ changes, the widely spaced residents along the line $p=0$ are displaced by the invasion of a mutant with $c>c_\textrm{Q}$ and $p>0$, converging to a nearly homogeneous cluster quite rapidly.

When once leaving the line $c=0$, monomorphic populations should move along the orbits determined by the fitness landscape for a rare mutant invading at the resident traits (called ``selection gradient'', $D(x)$ in Eq. (\ref{eq:2.8})), which revolve around the center Q. We remark that in the individual-based model the orbit the population travels on would first be stochastically selected, depending on the former invading mutant. Then, for orbits further out, the population will eventually be absorbed to the boundary $p=0$ after orbiting, and again, be exposed to the effect of the neutral drift. 

\item[(ii) \textit{Convergence to the center.}]

When orbits further in have been selected, the directional selection pressure becomes so weak that the effects of the non-zero correlation between the two traits are relatively considerable.
Cycles around the center Q can be observed in individual-based simulations for typical sets of parameters for the standard voluntary public good game \cite{HauEtal02a,HauEtal02b}. 
According to the numerical investigations, the population is certain to reach a sufficiently small vicinity of Q sooner or later, differently of the deterministic prediction by the canonical equation. 
In the case the population faces the next dispersal phase.

\begin{figure}
\begin{center}
\includegraphics[scale=0.80]{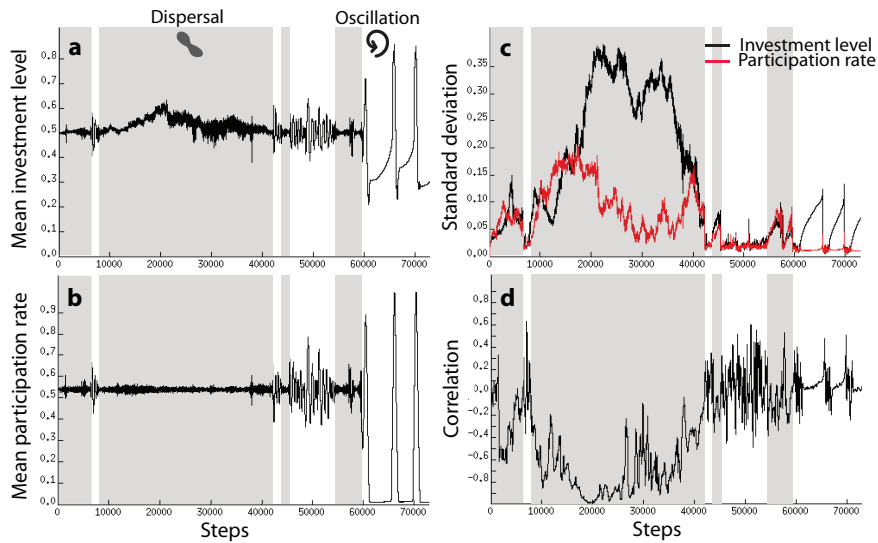}
\end{center}
\caption{\label{fig:6}
Evolutionary history consisting of oscillation and dispersal.
The initial state is monomorphic at the center Q.
Other parameters and panel features are as in Fig. \ref{fig:3}.
In this trial a quite large standard deviation in the investment level $c$ appears from 20,000 to 35,000 steps. As is shown in Fig. \ref{fig:7}, complete branching into both the extreme levels $c=0$ and $c=1$ ensues while the negative correlation between the two traits remains.         
}
\end{figure} 

\begin{figure}
\begin{center}
\includegraphics[scale=0.80]{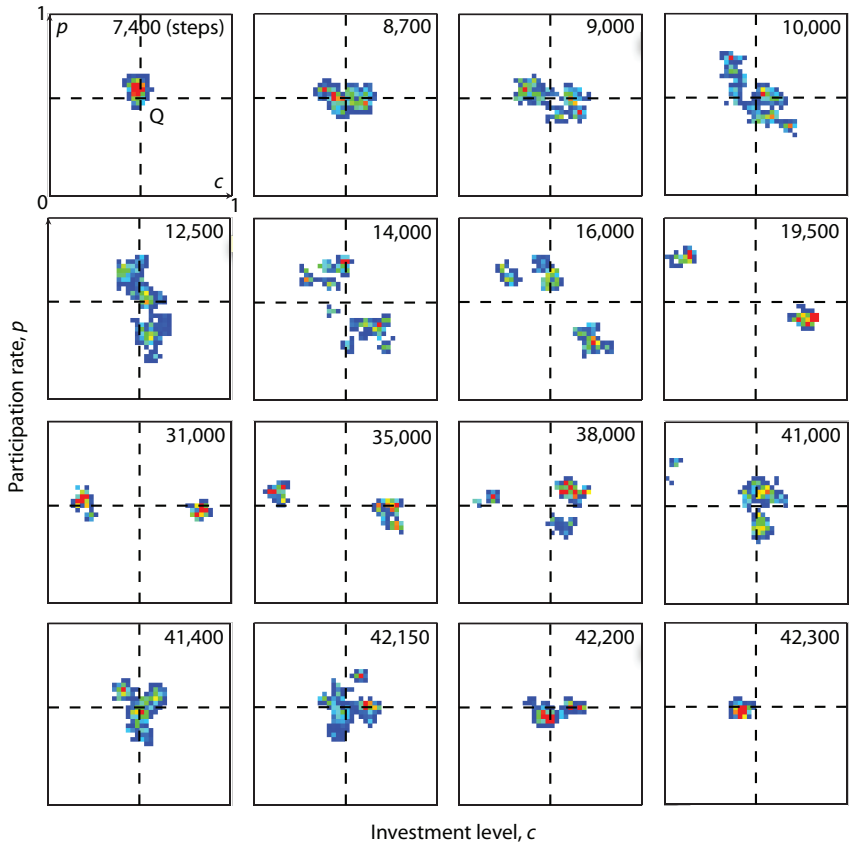}
\end{center}
\caption{\label{fig:7}
Snapshots of the dispersal phase corresponding to Fig. \ref{fig:6} (from 7,400 to 42,300 steps).
Panel features are the same as for Fig. \ref{fig:4}.
Two clusters that are distant from each other on either side of Q emerge after 12,500 steps, and then evolve into two distinct branches, each of which lies close to a pure strategy state, that is, full cooperation ($c=1$) or defection ($c=0$).
The right subcluster closer to Q becomes more widely distributed after 35,000 steps. 
After the left cluster vanishes (41,000 and 42,200 steps), the diversity quickly contracts, and then the population returns to monomorphism around Q (42,300 steps) (see also Fig. \ref{fig:5}).         
}
\end{figure} 

\item[(ii) \textit{Growth of the cloud and its collapse.}]

After converging to Q, the cluster starts radiating gradually. Some kinds of disturbance (e.g., arising from finite population sizes) cause biased spreading, such that even if it is only minor, minute fluctuations of the frequency distribution can occur, as observed in each dispersal phase (\textbf{a}, \textbf{b} of Figs. \ref{fig:3} and \ref{fig:6}). Although the outcomes of such fluctuations are difficult to predict analytically, the numerical simulations demonstrate that the diversified distribution survives when there has been a negative correlation between the two traits $c$ and $p$, or otherwise, does not.

Indeed, the first dispersal phase in Fig. \ref{fig:3} shows that the standard deviations repeat a small growth-and-decline cycle simultaneously, but neither has developed to a significant level. Then in the second dispersal phase, there is distinct development of diagonally located clusters (see also Fig. \ref{fig:5}) with a negative slope. If the diagonal dispersal goes further as shown in Figs. \ref{fig:6} and \ref{fig:7}, the resulting population can branch into two distinct clusters quite distant from each other across Q and there is long-run maintenance of a negative correlation over more than 30,000 steps.

The minute fluctuations remain throughout the dispersal phase, and are likely to affect clusters relatively closer to Q. Such an affected cluster then becomes more dispersed (sometimes splitting into subclusters). As a result, the local trend that is restricted to such a dispersed subgroup may be positive while the global trend still remains negative. This can cause an increase in the amplitude of fluctuations and the instability of the global population dynamics, often leading to extinction of some clusters. One can observe bursts of fluctuations that result in the diversity of the distribution being lost, as shown in Figs. \ref{fig:3} and \ref{fig:6}. However, if a cluster that vanishes has played a major part in maintaining the negative trend overall (see Fig. \ref{fig:5}, 114,490 steps; Fig. \ref{fig:7}, 41,000 and 42,200 steps), then there can be an incidental change to a positive trend. The positive correlation leads the diversified population to abruptly contract, often to a nearly uniform state. The location to which the distribution has contracted is often so far from the center Q that directional selection drives the population instead of neutral drift. We thus return to the cyclic oscillation phase.
\end{description}

We have also checked that the evolutionary history with both oscillation and dispersal as in Fig. \ref{fig:3} can happen for other parameters close to those for Fig. \ref{fig:3}, including large exploration probability 0.05.
We note that the dispersal phase is subdued compared to cases with high exploration probability $\mu$ or large standard deviation of explorations $s$. By contrast, the oscillation phase is largely unaffected by changes in the exploration parameters $\mu$ and $s$.

\section{Analysis}

Two complementary approaches are used for studying evolutionary dynamics of continuous games: replicator dynamics on probability distributions of strategies \cite{Oech02} and the canonical equation of adaptive dynamics \cite{dieckmann}. In the former approach, all strategies are present in the population at all times, while in the latter only certain resident strategies are present and other strategies are tried out occasionally through tinkering with investment and participation levels.  
In this section we use the mathematical framework of adaptive dynamics \cite{dieckmann,geritz,HofSig98,Ake13} to analyze the evolution resulting from the individual-based model. 

\subsection{Cyclic oscillation}

\subsubsection*{Invasion fitness}

We consider an infinitely large population and assume that at regular intervals $N$ individuals with $N \ge 2$ are randomly selected and offered the option to participate in the public good game. 
Consider a monomorphic resident in which each individual uses the same strategy $x=(c_x,p_x)$. Let us introduce \textit{invasion fitness} $S(x,y)$, that denotes the growth rate of a mutant strategy $y=(c_y,p_y)$, in the monomorphic resident. We assume that the growth rate of a rare mutant in a resident population is determined by the replicator dynamics. We can thus obtain an expression for the invasion fitness for the rare mutant within the resident population as follows:
\begin{equation}
\label{eq:2.1}
S(x,y)=P(x,y)-\bar{P}(x),
\end{equation}
where $P(x,y)$ describes the expected payoff of the mutant player with $y$ interacting with the other $N-1$ resident players with $x$, and $\bar{P}(x)$ expresses the average payoff over the resident population (and thus $\bar{P}(x)=P(x,x)$).
In the model it is convenient to define by $g(x,y)$ the mutant's expected payoff when the mutant participates in the public good game with its own probability $p_y$. Using this and the nonparticipant's payoff $\sigma$, the mutant's expected payoff is given by a linear equation with respect to $p_y$:
\begin{equation}
\label{eq:2.2}
P(x,y)=p_y g(x,y) + (1-p_y) \sigma.
\end{equation} 
See Eq. (\ref{eq:2.3}) in Appendix A.1 for details of $g(x,y)$. This yields
\begin{eqnarray}
\label{eq:2.5}
P(x,y) \! &=& \! \sigma \!+\! (r\!-\!1)(c_x \!-\! c_\textrm{Q})(1\!-\!(1\!-\!p_x)^{N-1}) p_y \!-\! (c_y \!-\! c_x) p_y F(1\!-\!p_x),
\end{eqnarray}
in which  $c_\textrm{Q}=\sigma/(r-1)$ and
\begin{eqnarray}
\label{eq:2.4}
F(z) = 1+(r-1)z^{N-1}-\frac{r(1-z^N)}{N(1-z)}.
\end{eqnarray}
We note that the set of roots of $F(z)$ in $[0,1]$ is only $z=1$ for $r\leq2$, and in addition there is a unique interior root $z_\textrm{Q}$ for $r>2$ \cite{HauEtal02b}.
Then we find that the invasion fitness in Eq. (\ref{eq:2.1}) is linear in the mutant traits $c_y$ and $p_y$, as follows,
\begin{equation}
\label{eq:2.7}
S(x,y)=a_{10}(x)(c_y-c_x)+a_{01}(x)(p_y-p_x)+a_{11}(x)(c_y-c_x)(p_y-p_x),
\end{equation}
where $a_{10}(x)=-p_{x}F(1-p_{x})$, $a_{01}(x)=(r-1)(c_{x}-c_Q)(1-(1-p_{x})^{N-1})$, and $a_{11}(x)=-F(1-p_{x})$.

\subsubsection*{Selection gradient}

The adaptive dynamics of the resident strategy $x$ is governed by the selection gradient as below, with the exception of the vicinity of its equilibrium points.
\begin{equation}
\label{eq:2.8}
D(x)
= \left( \genfrac{}{}{0pt}{}{ \left. \frac{\partial S(x,y)}{\partial c_y} \right|_{y=x} }
{ \left. \frac{\partial S(x,y)}{\partial p_y} \right|_{y=x}} \right)
= \binom{a_{10}(x)}{a_{01}(x)},
\end{equation}
so that $\dot{x}=D(x)$. An example is given in Fig. \ref{fig:1}.
This vector associated with strategy $x$ points in the direction of the maximal increase of the mutant's advantage over the resident population: $D(x)$ suggests the most favorable direction. 
In general, the adaptive dynamics for monomorphic populations with strategy $x$ is expressed using its canonical equation $\dot{x}= kVD(x)$, where $V$ is the variance-covariance matrix of the difference vector between the mutant and its parent, and the coefficient $k$ depends on the equilibrium for the population size and the mutational process \cite{dieckmann,Meszena,Lei09}. From the assumptions of the individual-based model, hereafter we analyze the canonical equation with $kV=\mathbf{1}$ (unit matrix).
The system being considered is
\begin{eqnarray}
\label{eq:2.9}
\dot{c} &=& -p F(1-p), \\
\label{eq:2.10}
\dot{p} &=& (r-1)(c-c_Q)(1-(1-p)^{N-1}),
\end{eqnarray}
which result in that every orbit has the line of symmetry $c=c_\textrm{Q}$.
We focus on the upper half plane $p>0$, excluding the exceptional line $p = 0$, a continuum of equilibria.
In the case $r\le 2$ (Fig. \ref{fig:1}\textbf{b}), this plane is filled with the orbits issuing from the one half of the $p$-axis $\{ c>c_\textrm{Q}, p=0 \}$, moving counter-clockwise around the point $(c_{\rm{Q}},0)$, and converging to the other half $\{ c<c_\textrm{Q}, p=0 \}$.
When an orbit converges to the boundary of the strategy space, we assume that the orbit then is governed by the selection gradient naturally projected to the boundary. For instance, in Fig. \ref{fig:1}\textbf{b} an orbit contacting the line $c = 0$ will move along this line with decrease in $p$, converging to $p=0$.     

We then turn to the case that $r > 2$ (Fig. \ref{fig:1}\textbf{a}) (see \cite{HauEtal02b} for the discrete voluntary public good game).
Dividing the right-hand sides in Eqs. (\ref{eq:2.9}) and (\ref{eq:2.10}) by $1-(1-p)^{N-1}$, which corresponds to a change of velocity, thus does not affect the orbits on the subspace. This yields 
\begin{eqnarray}
\label{eq:2.11}
\dot{c} &=& - \frac{pF(1-p)}{1-(1-p)^{N-1}}
=  - \frac{F(1-p)}{1+\sum_{i=1}^{N-2}{(1-p)^{i}}}, \\
\label{eq:2.12}
\dot{p} &=& (r-1)(c-c_Q).
\end{eqnarray}
We define these as $-g(p)$ and $l(c)$, respectively. We then introduce $H(c,p):=G(p)+L(c)$, where $G(p)$ and $L(c)$ are primitive functions of $g(p)$ and $l(c)$, respectively. The function $H$ is a constant of motion: $\dot{H}=\frac{\partial H}{\partial c} \dot{c} + \frac{\partial H}{\partial p} \dot{p} \equiv 0$.
The Hessian of $H$ is symmetric positive definite (and thus $H$ attains a strict minimum) at  the point ${\rm{Q}}=(c_{\rm{Q}}, p_{\rm{Q}})$, where $p_{\rm{Q}}$ equals $1-z_{\rm{Q}}$. We have used the fact that $F(z_{\rm{Q}})=0$ and  $F'(z_{\rm{Q}})<0$ \cite{HauEtal02b}. Therefore, a neighborhood of Q is filled with closed periodic orbits.

To see the global dynamics, we explore the vicinity of the boundary point $(c_\textrm{Q},0)$. The Jacobian for Eqs. (\ref{eq:2.11}) and (\ref{eq:2.12}) at this point is given by
\begin{equation}
\label{eq:2.15}
\left (
\begin{matrix}
0 & \frac{F'(1)}{N-1} \\
r-1 & 0
\end{matrix}
\right ).
\end{equation}
We note that $F'(1)=(\frac{r}{2}-1)(N-1) > 0$ for $r>2$, in which case it follows that the matrix has two real eigenvalues of different sign and thus the equilibrium $(c_{\rm{Q}},0)$ is a saddle point. 
For $p>0$, the orbits in the vicinity of the saddle point agrees with the orbits associated with Eqs. (\ref{eq:2.9}) and (\ref{eq:2.10}).  
Moreover, considering the symmetry to the line $c=c_{\rm{Q}}$, the separatrices for the saddle point connect each other and comprise the critical level set surrounding Q; its inside is filled with the closed periodic orbits, and its outside is filled with heteroclinic orbits issuing from the one half of the line $\{ c>c_\textrm{Q}, p=0 \}$, turning around Q, and returning to the other half $\{ c<c_\textrm{Q}, p=0 \}$.    
   
\subsubsection*{Singular strategies}

Equilibrium points of the selection gradient are called singular strategies, and they are given by solutions of $D(x)=0$.
If the monomorphic population takes a singular strategy, the selection pressure on both the directions of $c$ and $p$ vanishes; otherwise, the monomorphic population is always under directional selection. 
From Eqs. (\ref{eq:2.9}) and (\ref{eq:2.10}), it follows that singular strategies of the continuous voluntary public good game are given by the boundary line $p = 0$, and for $r>2$, also the point Q.
The dynamic analysis implies that Q is a center surrounded locally by periodical closed orbits \cite{hirsch}, and as long as mutants arise infinitesimally rare and near, the monomorphic population on the cyclic orbit should not reach the center Q.

For  $r>2$, the $c$-coordinate of Q, $c_{\rm{Q}}=\sigma/(r-1)$, increases with the nonparticipant's payoff $\sigma$ and decreases with increasing the multiplication factor $r$. Changes in  the group size $N$ do not affect the $c$-coordinate of Q. In the case of its $p$-coordinate, $p_{\rm{Q}}$, we know from the differential (or difference) of $F$ in Eq. (\ref{eq:2.4}) with respect to $r$ or $N$, that $F(1-p)$ decreases with increasing $r$ and increases with $N$. Considering that $F(1-p)$ is positive for $p_{\rm{Q}}<p \le 1$ or negative for $0<p<p_{\rm{Q}}$ \cite{HauEtal02b}, thus the value of unique interior root $p_{\rm{Q}}$ must increase with $r$ and decrease with increasing $N$.   

If $x$ is a singular strategy among Q and $p=0$, the invasion fitness degenerates, that is, $S(x,y)=P(x,y)-P(x,x)=0$ holds for all strategy $y$ in the strategy space $U$. (We note that $P(x,y)=\sigma=P(x,x)$ if $x = \textrm{Q}$ or $p_x=0$.)
This implies that Q and each point of $p=0$ are (symmetric and not strict) Nash equilibria and from the equality any mutant strategy may first through neutral drift invade a resident population with the singular strategy. 
In Appendix A.3 we fully analyze the replicator dynamics for two distinct strategies in the continuous voluntary public good game.  
Indeed, we have that $S(\textrm{Q},u) = 0$ and $S(u,\textrm{Q}) < 0$ for all strategy $u$ within $\{ c<c_\textrm{Q}, p>p_\textrm{Q} \}$ or $\{ c>c_\textrm{Q}, 0< p<p_\textrm{Q} \}$. Considering Eq. (\ref{eq:A.34}), this indicates that the strategies Q and $u$ are mutually invasible in the replicator dynamics.
In the case of a resident population with $p=0$, then any mutant $y$ within $ \{ c > c_\textrm{Q}, p>0 \}$ can invade and eventually replace the resident nonparticipation strategy. Therefore, neither Q nor $p=0$ is an evolutionary stable strategy (ESS) \cite{may73,Less90} and thus satisfy any of the following properties: evolutionary robust strategy \cite{Oech02}, strongly uninvadable strategy \cite{Bom90}, and continuously stable strategy (CSS) \cite{Esh83}.
In addition, if its initial fraction $\varepsilon$ is sufficiently small, the mutant strategy Q is not able to invade any nearby resident population in $\{ c>c_\textrm{Q}, p>p_\textrm{Q} \}$ or $\{ c<c_\textrm{Q}, 0< p<p_\textrm{Q} \}$.
Also considering that neutrality among strategies with $p=0$, this implies that neither Q nor $p=0$ is an neighborhood invader strategy (NIS) \cite{Apo97}. The singular strategies above are not always effective to invade a nearby resident population.

\begin{figure}
\begin{center}
\includegraphics[scale=0.68]{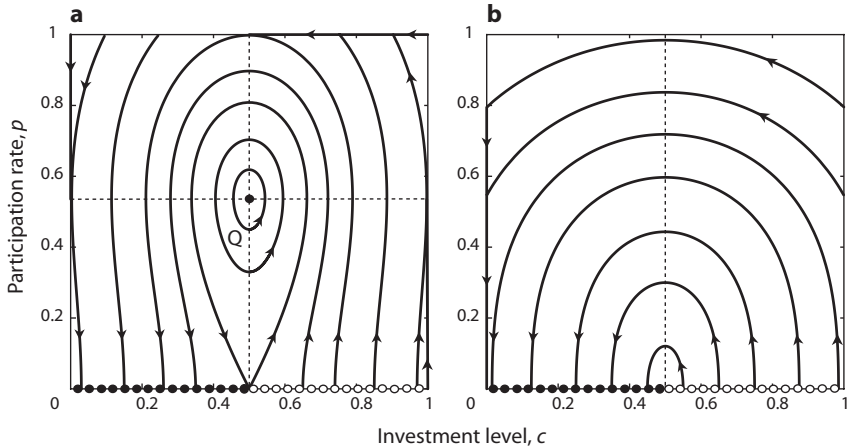}
\end{center}
\caption{\label{fig:1}
Selection gradient depicted by Eqs. (\ref{eq:2.9}) and (\ref{eq:2.10}) for $N=5$, and in \textbf{a}, $r=3$ and $\sigma=1$, or in \textbf{b}, $r=1.8$ and $\sigma=0.4$. A set of singular strategies consists of the point $\textrm{Q}(0.5,p_\textrm{Q})$ and the line $p=0$ ($p_\textrm{Q} \approx 0.5387$) in \textbf{a} or only the line in \textbf{b}. 
(\textbf{a}) The point Q is a center surrounded by closed orbits. The point $(0.5,0)$ divides the boundary line $p=0$ into a left segment of stable fixed points (Lyapunov stable; \textit{closed} circles) and a right segment of unstable fixed points (\textit{open} circles).         
(\textbf{b}) The strategy space $U$ is filled with the orbits connecting the right segment of $p=0$ (or the line $c=1$) and left segments of $p=0$.
}
\end{figure} 

When the population is monomorphic with strategy Q, then the canonical equation, based on infinitesimal small mutations, is not able to predict the evolutionary dynamics, because of the flatness of invasion fitness ($S(\textrm{Q},y)=0$ for all $y$ in $U$) and substantial effects of neutral drift, instead of directional selection.   
Indeed, with considering a more realistic situation in that exploration is not always small and emergent innovation happens \cite{Hof09}, evolutionary dynamics can more easily result in branching and polymorphism \cite{geritz2,kisdi,vukics,Nishimura,ito,ito2}. 
For instance, in a sufficiently small neighborhood of Q directional selection is so weak that a next exploration would occur before substitution of the resident strategy has been completed. For last decade fundamental techniques to investigate convergence and stability for multi-dimensional adaptive dynamics have been developed \cite{Cress05,Cress06}.  
Yet satisfactory analyses of evolutionary diversification around a non-ESS singular strategy, in particular like the center Q, have not been published. In the following section, accordingly, to investigate effects on cyclic orbits of small, but finite mutations, we will expand the analysis of adaptive dynamics to the \textit{neutral} direction along normal vectors of the selection gradient.

\subsection{Convergence to the center}

Here we will provide a geometrical analysis of local evolutionary dynamics and classify the strategy space $U$ into four regions, depending on the abilities to invade and to be invaded for a given strategy \cite{mazancourt}. 
Of particular interest is to explore whether there is a meaningful trait region in which evolutionary branching can happen. 
Coupling these abilities generates a two-dimensional pairwise invasibility plot \cite{geritz}, clarifying what type of evolutionary scenarios happen between minor mutants and major residents.

From the arrangement of the singular strategies and isoclines, the subspace of $U$, in which both the directional selection pressure do not vanish, are naturally divided to four types of quadrant around Q, \textbf{I} $=\{ c>c_\textrm{Q}, p>p_\textrm{Q} \}$, \textbf{II} $=\{ c<c_\textrm{Q}, p>p_\textrm{Q} \}$, \textbf{III} $=\{ c<c_\textrm{Q}, p<p_\textrm{Q} \}$, and \textbf{IV} $=\{ c>c_\textrm{Q}, p<p_\textrm{Q} \}$.  
Let us look the curve $C_m = \{ (c_y,p_y )| S(x,y)=0 \}$, which goes through the given focal point $x$ and separates regions of (mutational) strategies that can invade into the resident population with $x$ from strategies that cannot.
Equation (\ref{eq:2.7}) yields that $C_m$ is a hyperbolic curve for all points $x$ in $U$ except $x=\textrm{Q}$, as follows:
\begin{eqnarray}
\label{eq:2.16}
p_y \left[
c_y - \left( c_x - \frac{a_{01}(x)}{a_{11}(x)} \right)
\right]
= p_x \left( \frac{a_{01}(x)}{a_{11}(x)} \right)
\begin{cases}
<0 & \text{if \ $x \in$ \bf{I} $\cup$ \bf{III}},\\
>0 & \text{if \ $x \in$ \bf{II} $\cup$ \bf{IV}}.
\end{cases}
\end{eqnarray}
The asymptotic lines of the hyperbola $C_m$ are thus given by $p_y=0$ and $c_y=c_x-a_{01}(x)/a_{11}(x)$.
The asymptotic line $p_y=0$ means that all we need to consider is the upper one of two connected components of the hyperbola, because the lower one is located entirely out of $U$.
Then, let us consider the curve $C_r = \{ (c_y,p_y )| S(y,x)=0 \}$, which separates (resident) strategies that can be invaded by a given mutant with $x$ from strategies that cannot.
It is obvious that $C_r$ goes through $x$ and Q.
We analyze $C_r$ locally in the neighborhood of $x$. 
From a straightforward calculation, we can obtain that $C_r$ and $C_m$ have identical slopes at $x$.
So, to clarify the local arrangement of $C_m$ and $C_r$, there remains only an investigation of these second order approximations because $C_m$ is a quadric curve. Appendix A.2 provides calculations of curvatures $\kappa_m(x)$ and $\kappa_r(x)$ of $C_m$ and $C_r$ at $x$, respectively, and estimations of signs of these curvatures and the difference between them. Signs of all indexes we need to consider are displayed on Table \ref{table:1} and organized by the four quadrants \textbf{I}, \textbf{II}, \textbf{III}, and \textbf{IV}.     

\begin{table}[tb]
\caption{Sign configurations of indexes for geometric analysis}\label{table:1}
\begin{center}
\begin{tabular}{c|cccc}
regions & \bf{I} & \bf{II} & \bf{III} & \bf{IV} \\ \hline
$a_{01}(x)$ & $+$ & $-$ & $-$ & $+$ \\
$a_{11}(x)$ & $-$ & $-$ & $+$ & $+$ \\
slope: $p'_y(x)$ & $+$ & $-$ & $+$ & $-$ \\
$\kappa_m(x)$ & $+$ & $+$ & $+$ & $+$ \\
$\kappa_r(x)$ & $-$ & $-$ & $+$ & $+$ \\
$\kappa_r(x)-\kappa_m(x)$ & $-$ & $-$ & $+$ & $+$ \\ 
\end{tabular}
\end{center}
\end{table}

\begin{figure}
\begin{center}
\includegraphics[scale=0.68]{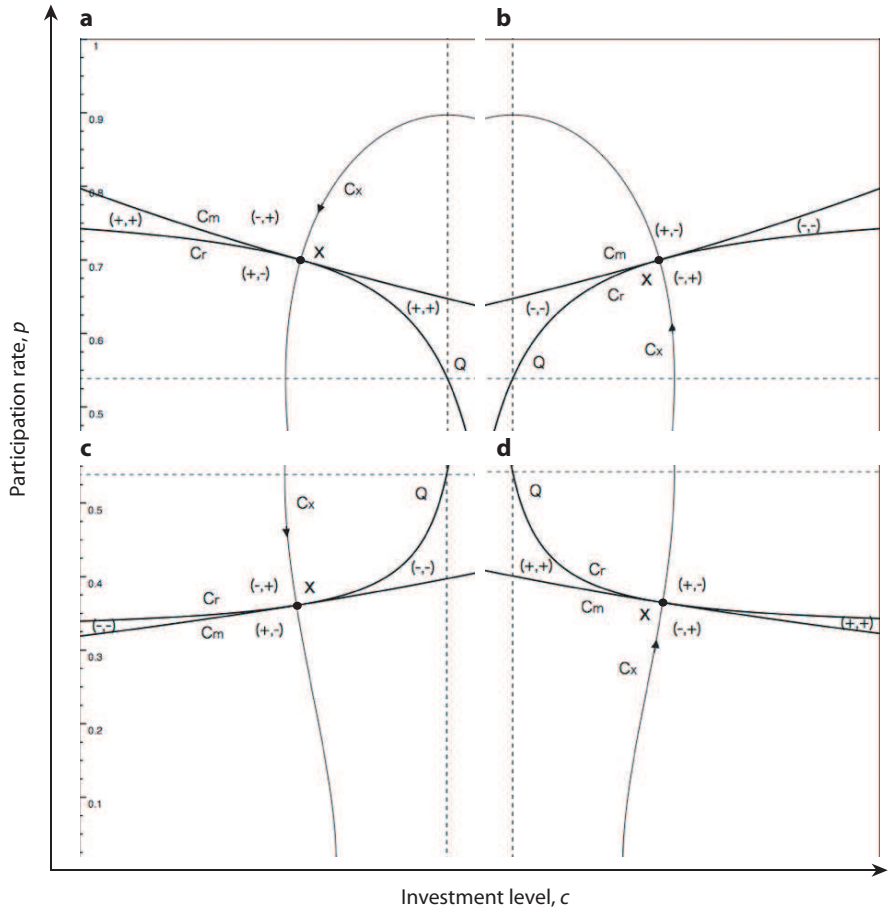}
\end{center}
\caption{\label{fig:2}
Geometric analysis of pairwise invasibility for $N=5$, $r=3$, and $\sigma=1$. The point ${\textrm{Q}}\,(0.5,p_\textrm{Q})$ is a center surrounded by cyclic orbits ($p_\textrm{Q} \approx 0.5387$). 
The focal points $x$ are given by $(0.7,0.7)$, $(0.3,0.7)$, $(0.2957,0.361)$, and $(0.7043,0.361)$ in \textbf{a}, \textbf{b}, \textbf{c}, and \textbf{d}, respectively.
These points ride on the same orbit $C_x$ and are located, respectively, in quadrants \textbf{II}, \textbf{III}, \textbf{IV}, and \textbf{I} around Q, which have different local arrangements of the two specific curves, $C_m:S(x,y)=0$ and $C_r:S(y,x)=0$.
The curves $C_m$ and $C_r$ are tangent to each other at the focal point $x$.
$C_m$ is a hyperbola, and $C_r$ passes through $\textrm{Q}$. 
These curves are orthogonal to the orbit $C_x$.
Sign couplings, e.g. $(+,-)$, means $(\text{sgn}\:S(x,y),\text{sgn}\:S(y,x))$.
For instance, a rare mutant with $y$ in the $(+,-)$-part is able to invade the resident population with the focal strategy $x$ (because $S(x,y)>0$); yet a rare mutant with $x$ is not able to invade the resident population with $y$ (because $S(y,x)<0$). According to Appendix A.3, this indicates that the focal resident $x$ can be replaced with any mutant in the the $(+,-)$-part.
}
\end{figure} 

From Table \ref{table:1}, we can obtain complete information about possible combinations of the signs of $S(x,y)$ and $S(y,x)$ in the vicinity of a focal point $x$.
Figure \ref{fig:2} provides typical plots of $C_r$ and $C_m$, and configurations of $(\text{sgn}\:S(x,y),\text{sgn}\:S(y,x))$.
According to Appendix A.3, in the continuous voluntary public good game the type of frequency dynamics between any two strategies $x_1$ and $x_2$ can be determined by $(\text{sgn}\:S(x_2,x_1),$ $\text{sgn}\:S(x_1,x_2))$. 
As such, the case of $(\text{sgn}\:S(x,y),\text{sgn}\:S(y,x))=(+,-)$ means that the mutant with $y$ can invade the resident population with $x$ and inevitably replace it. 
On the other hand, the case of $(-,+)$ means that the mutant $y$ cannot invade the resident $x$ and any mixed state of strategies $x$ and $y$ leads to extinction of $y$.
The case of $(+,+)$ means that strategies $x$ and $y$ can invade mutually and there is uniquely a stable coexisting state.
The case of $(-,-)$ means that strategies $x$ and $y$ cannot invade each other, and thus, the monomorphic resident population with $x$ is stable if the number of mutants is infinitesimally small.
For a given point $x$ under directional selection, the neighborhood of $x$ always has two types of combinations, $(+,-)$ and $(-,+)$, in the front of and behind $x$ along the direction of evolution, respectively.
    
Considering the signs of curvatures of $C_m$ and $C_r$, local evolutionary dynamics can be classified to four types, corresponding to the four regions \textbf{I}, \textbf{II}, \textbf{III}, and \textbf{IV}, as in the following list. Let us denote as $C_x$ the orbital curve at $x$, and as $N_x$ the line orthogonal to $C_x$ at $x$.
\begin{description}
\item[\textbf{I}]
$C_m$ is concave and $C_r$ is convex to the direction of evolution along the orbit $C_x$. $C_r$ is located behind $C_m$.
We have $(\text{sgn}\:S(x,y),\text{sgn}\:S(y,x))=(-,-)$ at the gap between $C_m$ and $C_r$.
$N_x$ is in the gap and sandwiched by these curves.
\item[\textbf{II}]
$C_m$ is convex and $C_r$ is concave to the direction of evolution along $C_x$.
$C_r$ is located in the front of $C_m$.
We have $(+,+)$ at the gap between $C_m$ and $C_r$.
$N_x$ is included in the gap.
\item[\textbf{III}]
Both $C_m$ and $C_r$ are convex to the direction of evolution along $C_x$. 
$C_r$ is located behind $C_m$.
We have $(-,-)$ at the gap between $C_m$ and $C_r$.
$N_x$ is located in the front of both $C_m$ and $C_r$ and within $(+,-)$-sign part.
\item[\textbf{IV}]
Both $C_m$ and $C_r$ are concave to the direction of evolution along $C_x$. 
$C_r$ is located in the front of $C_m$.
We have $(+,+)$ at the gap between $C_m$ and $C_r$.
$N_x$ is located behind both $C_m$ and $C_r$ and within $(-,+)$-part.
\end{description}

\subsection{Growth of the cloud and its collapse}

\subsubsection*{Invasion fitness for polymorphic populations}

We focus on stable dimorphism arising by $(+,+)$-part in \textbf{II} or \textbf{IV}. 
A mutant arising from the $(+,+)$-part for the focal point $x$ is invasible to and can coexist with $x$.
So far, it has been assumed that the focal resident population is monomorphic.
According to Appendix A.1, polymorphic populations with strategy distribution $X$ are naturally mapped on monomorphic populations with strategy $x^*$ while the mapping $\pi$ keeps the invasion fitness unchanged: 
\begin{equation}
\label{eq:2.18}
S(X,y)=S(x^*,y).
\end{equation}
The mapping $x^*=(c_{x^*},p_{x^*})=\pi(X)$ is given by that $p_{x^*}=\bar{p}_x$ and $c_{x^*}=\bar{c}_x + \textrm{Cov}(X)/\bar{p}_x$ in Eqs. (\ref{eq:A.7}) and (\ref{eq:A.8}). 
Equation (\ref{eq:2.18})  means that arguments for monomorphic resident populations can extend to polymorphic resident populations by using the representative strategy $x^*$.
Coexisting patches of $X$ must then be situated on the hyperbola $C_{r^*} = \{ (c_y,p_y)| S(x^*,y)=0 \}$ associated with $x^*$; in particular, in the case of two patches, these are located on $C_{r^*}$ across $x^*$.

\subsubsection*{Selection gradient for polymorphic populations}

To understand in depth how strategic diversification can affect the evolutionary fate of populations, we would need to extend the selection gradient to polymorphic resident populations.   
We assume that a rare mutant is stochastically emerging around each patch of a resident population and the mutant's probability distribution is proportional to the frequency distribution of the resident population.
Let $X=\{ (x_i, n_i) \}_{i=1,\cdots, K}$ be the distribution 
of a polymorphic resident population with finite support,
where $x_i=(c_{x_i},p_{x_i})$ and each $n_i$ denotes the relative frequency of strategy $x_i$ with $\sum_{1\le i \le K} n_i = 1$.
We then employ the mutant's distribution $Y$ related to $X$, 
such that $Y=\{ (y_i, n_i) \}_{i=1,\cdots, K}$ where any $y_i(=(c_{y_i},p_{y_i}))$ is sufficiently close to $x_i$.

To predict evolutionary trajectories after diversification, we define an integrated selection gradient for the averaged strategy of the resident distribution $X$, $\bar{x}=(\bar{c}_{x},\bar{p}_{x})$, 
as a weighted sum of the selection gradient induced by each mutant strategy $y_i$ arising near $x_i$, that is,   
\begin{equation}
\label{eq:2.19}
	D(\bar{x})
	:= \left( \genfrac{}{}{0pt}{}{ \sum^K_{i=1} \left. n_i\frac{\partial S(X,y_i)}{\partial c_{yi}} \right|_{y_i=x_i} }
	{ \sum^K_{i=1} \left. n_i\frac{\partial S(X,y_i)}{\partial p_{yi}} \right|_{y_i=x_i}} \right),
\end{equation}
where
\begin{eqnarray}
\label{eq:2.20}
{\left . \genfrac{}{}{}{0}
{\partial S(X,y_i)}{\partial c_{y_i}} \right |_{y_i=x_i}}
&=& a_{10}(x^*)+a_{11}(x^*)(p_{x_i}-p_{x^*}), \\
\label{eq:2.21}
{\left . \genfrac{}{}{}{0}
{\partial S(X,y_i)}{\partial p_{y_i}} \right |_{y_i=x_i}}
&=& a_{01}(x^*)+a_{11}(x^*)(c_{x_i}-c_{x^*}). 
\end{eqnarray}
Then, Eq. (\ref{eq:2.19}) is rewritten as
\begin{equation}
\label{eq:2.22}
	D(\bar{x}) =
	\left(
	\begin{matrix}
	a_{10}(x^*)  \\
	a_{01}(x^*) - a_{11}(x^*) \displaystyle{\frac{\textrm{Cov}(X)}{\bar{p}_x}}
	\end{matrix}
	\right)
        =
	\left(
	\begin{matrix}
	a_{10}(\bar{x})  \\
	a_{01}(\bar{x}) + G(\bar{z}_{x}) \textrm{Cov}(X)
	\end{matrix}
	\right),
\end{equation}
where $\bar{z}_{x}=1-\bar{p}_x$, and
\begin{eqnarray}
\label{eq:2.23}
G(z) = \frac{r}{1-z} \left(1- \frac{1-z^N}{N(1-z)} \right)
> 0 \quad \text{for all $z$ in $[0,1]$}.
\end{eqnarray}
If $X$ is monomorphic, $\textrm{Cov}(X)=0$, and hence, Eq. (\ref{eq:2.8}) is recovered. 
Monomorphic and polymorphic populations differ in the additional terms, $-a_{11}(x^*)\textrm{Cov}(X)/\bar{p}_x$ and $G(\bar{z}_{x})\textrm{Cov}(X)$ in Eq. (\ref{eq:2.22}), which would lead to some changes in orbits along which the average strategy $\bar{x}$ travels.

We have analyzed in Sect. 4.2 that there is a chance of polymorphism along a normal of the selection gradient in the quadrant \textbf{II} or \textbf{IV}. Thus, the slopes of the normals are negative, and also is the resulting covariance between the two traits $\textrm{Cov}(X)$. This leads to such that a diversified population (and its average strategy) in \textbf{II} and \textbf{IV} can turn in and out of those former orbits, respectively, as in Eq. (\ref{eq:2.22}). In the case of the quadrant \textbf{I} or \textbf{III}, there is no $(+,+)$-part and thus a rare possibility to maintain dimorphism that resulted from mutations. As such, it is not required to consider the effects of the additional term in Eq. (\ref{eq:2.22}) in \textbf{I} and \textbf{III}.
From these arguments, we could state that there is a specific region of the strategy space in which populations tend to be with the negative covariance and are likely to be pulled into inner orbits. 

\subsubsection*{Singular strategies for polymorphic populations}

From Eq. (\ref{eq:2.22}) we understand that Q is an equilibrium point even if $X$ is polymorphic.
Here we investigate the stability of the population with $x^*=\rm{Q}$.
We calculate the Jacobian of the selection gradient $D(\bar{x})$ in Eq. (\ref{eq:2.22}) with respect to $(c_{x^*},p_{x^*})$.
Considering that $\textrm{Cov}(X)/\bar{p}_{x} = c_{x^*}-\bar{c}_x$ is independent of $p_{x^*}$, $F(1-p_{\rm{Q}})=0$, and $(r-1)c_{\rm{Q}}-\sigma = 0$, we obtain
\begin{equation}
\label{eq:2.24}
J |_{x^*=\rm{Q}} 
=
 \left (
 \begin{matrix}
  0  &  p_{\rm{Q}} F'(1-p_{\rm{Q}}) \\
  (r-1)(1-(1-p_{\rm{Q}})^{N-1}) &  \textrm{Cov}(X) \left( \genfrac{}{}{}{0}{-F'(1-p_{\rm{Q}})}{p_{\rm{Q}}} \right)
    \end{matrix}
\right ).
\end{equation}
\begin{theorem}
The center $\rm{Q}$ is stable if and only if $\rm{Cov}(X)<0$.
\end{theorem}
\begin{proof}
Since $0<p_{\rm{Q}}<1$ and $F'(1-p_{\rm{Q}})<0$, 
the off-diagonal elements of $J |_{x^*=\rm{Q}}$ are of opposite sign and the remaining diagonal element has the same sign as $\textrm{Cov}(X)$.
Thus, we have that $\det J |_{x^*=\rm{Q}} > 0$ and that the sign of tr$J |_{x^*=\rm{Q}}$ is the same as $\textrm{Cov}(X)$.
The stability of the representative strategy at Q, $x^*=\rm{Q}$, is extensively analyzed, as follows: 
if $\textrm{Cov}(X)>0$, then tr$J |_{x^*=\rm{Q}} > 0$ and $\det J |_{x^*=\rm{Q}} > 0$, and hence, Q is unstable; 
otherwise, if $\textrm{Cov}(X)<0$, then tr$J |_{x^*=\rm{Q}} < 0$ and $\det J |_{x^*=\rm{Q}} > 0$, and hence, Q is stable. \qed
\end{proof}

In the individual-based simulations the symmetry of distribution is hardly maintained when considering mutations. 
After the symmetric distribution breaks, if $\textrm{Cov}(X)<0$, the representative strategy $x^*$ would stay at Q. Since $x^*=\rm{Q}$ leads to a situation where any mutant can invade the population by neutral drift, the population is likely to continue spreading as long as the negative covariance is maintained. This process would lead the distribution to evolve into a diagonally spread form with a negative slope, actually as is demonstrated by the numerical investigations. If $\textrm{Cov}(X)>0$, then $x^*$ will go away from Q. When $x^* \neq \rm{Q}$, cloud-like populations diversified through neutral drift would shrink into a shape determined such that $S(X,x_i)=S(x^*,x_i)=0$ for any patch $x_i$ of $X$.

\section{Discussion}

We investigate a model for understanding the evolution of cooperation among individuals in free and anonymous societies.
It is well known in discrete-strategy models that voluntary participation plays an important role in maintaining the level of cooperation \cite{HauEtal02a,HauEtal02b,SemEtal03,Akt04,MatBoy09,seki,IzqEtal10,CasTor10,XuEtal10,ZhongEtal13,Wu13} as well as in promoting the evolution of costly punishment for collective actions \cite{Fow05,BraEtal06,HauEtal07,SigEtal10,SigEtal11,SasEtal12,GarTra12,Sa13} (see also \cite{RanNow11}).
There are a handful of prior theoretical studies about the co-evolution of continuous investment in public good games and other continuous properties, such as costly punishment (severity, strictness, tolerance, etc.) \cite{NaDi09,ShiNa13}: in particular Shimao and Nakamaru \cite{ShiNa13} considered the evolution of graduated punishment in well-mixed populations and demonstrated the possibility of strategy diversification.
Yet, little has been known about the consequences of continuously varying participation probability. We have analyzed gradual co-evoution of participation probability in the continuous public good game and shown that the evolutionary dynamics is characterised three phases: osciallation, convergence, and unchecked strategy-diversficiation and collapse to a single strategy. 
 
Despite being a simple model set in well-mixed populsations, the continuous voluntary public good game is capable of a wide variety of evolutionary scenarios, roughly classified into oscillations (rock-scissors-paper cycle) and dispersal (emergence of the specific correlation among the traits).
We fully analyze the adaptive dynamics for monomorphic populations for the model. We remark that the dynamics obtained are qualitatively different of those for discrete voluntary public good games \cite{HauEtal02a,HauEtal02b} (see also \cite{Tra09} for effects of small populations in the discrete games). 
In the discrete case oscillations occur in population compositions (fractions of cooperators, defectors, and nonparticipants). However, in the continuous case oscillations occur in individual traits (levels of cooperation and participation if individuals).
The cyclic movement along periodical closed orbits could be described as in a rock-scissors-paper cycle: if most strategies are cooperative, it would be profitable to reduce investment; if more defective strategies are prevalent, it would be better to more frequently exit; if most individuals tend to be nonparticipants at most times, then small interaction groups form most frequently, in which case increasing in the cooperation and participation levels is more beneficial than otherwise.
We note that the theoretical results, depicted in Fig. \ref{fig:1}, continue to hold even if a single participant would benefit from the good provision, instead of the loner's payoff $\sigma$. See Appendix A.4 for details.

We have analyzed the unfolding evolutionary dynamics using the theory of adaptive dynamics. This worked well for understanding the initial oscillatory phase in which individual levels of cooperation and participation gradually converged to a singular center Q. At this point, the invasion fitness vanish entirely and all mutant strategies have the same expected payoff as the resident strategy. This can be understood as an instance of the Bishop-Cannings theorem \cite{BiCa} in evolutionary game theory. To understand the dispersal phase that ensues, we established a condition for mutual invasibility of nearby strategies and showed that strategy diversification can continue as long as the covariance of individual cooperation and participation levels is negative. While the process of diversification may continue for quite some time, random drift eventually results in imbalances and the manifestation of positive covariance. At this point, the cloud of coexisting strategies unexpectedly collapses and the process begins anew.

\begin{acknowledgements}
We thank Daniel Bisig, Ulf Dieckmann, Josef Hofbauer, Yoshiaki Kaneko, Karl Sigmund, and Kazuki Yamamoto as well as the referees and Editors. T.S. acknowledges support by the Foundational Questions in Evolutionary Biology Fund: RFP-12-21 and the Austrian Science Fund (FWF): P27018-G11. 
\end{acknowledgements}

\appendix

\renewcommand{\theequation}{A.\arabic{equation}}
\setcounter{equation}{0}

\section*{Appendix A}
 
\subsection*{A.1 \ Invasion fitness}
Here we first calculate the expected payoff of a mutant within a monomorphic or polymorphic resident population with finite support. We then introduce the invasion fitness of a mutant extended to being with a probability distribution.

\subsubsection*{For monomorphic populations}

The probability that the rare participating mutant (we assume, of the infinitesimally small frequency) finds no co-participant and $N-1$ nonparticipants from the resident population with $(c_x,p_x)$ is $(1-p_x)^{N-1}$. In the case by assumption the public good games does not hold and the single participant only earns $\sigma$. Or otherwise, the mutant player should find from the resident population $S -1$ co-participants and $N-S$ nonparticipants with $2 \le S \le N$, with probability $\binom{N - 1}{S - 1} p_x^{S-1} (1-p_x)^{N-S}$. 
Hence, the mutant will obtain the net benefit through the participation, $\frac{r}{S}((S-1)c_x+c_y)-c_y$. The payoff of the rare participating mutant is 
\begin{eqnarray}
\label{eq:2.3}
g(x,y) &=& (1-p_x)^{N-1} \sigma 
+ \sum_{S=2}^{N}{\binom{N\!-\!1}{S\!-\!1} p_x^{S-1} (1-p_x)^{N-S} \left[ \frac{r((S-1)c_x+c_y)}{S}-c_y \right] } \nonumber \\
&=& (1-p_x)^{N-1} \sigma  
+ \sum_{S=2}^{N}{\binom{N\!-\!1}{S\!-\!1} \frac{1}{S} p_x^{S-1} (1-p_x)^{N-S}  r (c_y - c_x) }  \nonumber \\
&& + \sum_{S=2}^{N}{\binom{N\!-\!1}{S\!-\!1} p_x^{S-1} (1-p_x)^{N-S} ( r c_x - c_y ) } \nonumber \\
&=& (1-p_x)^{N-1} \sigma 
+  \frac{1}{N} \sum_{S=2}^{N}{\binom{N}{S} p_x^{S-1} (1-p_x)^{N-S} r (c_y - c_x) }  \nonumber \\
&& + \sum_{S'=1}^{N-1}{\binom{N\!-\!1}{S'} p_x^{S'} (1-p_x)^{(N-1)-S'} ( r c_x - c_y ) } \nonumber \\
&=& (1-p_x)^{N-1} \sigma
+ \left( \frac{1- (1-p_x)^{N}}{N p_x} - (1-p_x)^{N-1} \right) r (c_x - c_y)  \nonumber \\
&& + (1-(1-p_x)^{N-1})  (r c_x - c_y)  \nonumber \\
&=& (1-p_x)^{N-1} \sigma + (1-(1-p_x)^{N-1}) (r-1)c_x - (c_y - c_x) F(1-p_x),
\end{eqnarray}
where $\binom{N\!-\!1}{S\!-\!1} \frac{1}{S} = \binom{N}{S} \frac{1}{N}$ and $F(z)$ is as in Eq. (\ref{eq:2.4}). This yields the mutant's payoff $P(x,y)$ in Eq. (\ref{eq:2.5}) and the invasion fitness $S(x,y)$ in Eq. (\ref{eq:2.7}).

\subsubsection*{For polymorphic populations}

We consider a polymorphic resident population consisting of different $K$ patches ($1 \le K < \infty$), denoted as $X=\{ (x_i, n_i) \}_{i=1,\cdots, K}$,
where $x_i$ represents a two-dimensional strategy given by $x_i=(c_{xi},p_{xi})$ and $n_i$ represents the relative frequency of strategy $x_i$.
Thus, $n_i \ge 0$ and $\sum^{K}_{i=1} n_i = 1$. 
Let us, in general, denote by $P(X,y)$ the expected payoff of an individual with strategy $y=(c_y,p_y)$ interacting with other $N-1$ individuals randomly selected from the polymorphic resident population $X$.
As well, $g(X,y)$ denotes the individual's expected payoff when the individual participates with probability $p_y$.
Hence,
\begin{equation}
\label{eq:A.1}
P(X,y)=p_y g(X,y) + (1-p_y) \sigma,
\end{equation}
As well as in the case of monomorphic populations, thus, the invasion fitness for the rare mutant within the polymorphic resident population is defined as
\begin{equation}
\label{eq:A.1}
S(X,y)=P(X,y)-\bar{P}(X),
\end{equation} 
in which $\bar{P}(X)$ describes the average payoff over the polymorphic resident population, given by $\sum^{K}_{i=1} n_i P(X,x_i)$. 

We next estimate the probability that the focal participant with $y=(c_y,p_y)$ finds from the resident population $S -1$ co-participants and $N-S$ nonparticipants.
In the case $S=1$ (no resident participates), the focal participant obtains the payoff $\sigma$ by assumption.
This happens with probability $(\sum^{K}_{i=1} n_i (1-p_{xi}))^{N-1}=(1-\bar{p}_x)^{N-1}$, where $\bar{p}_x:=\sum^K_{i=1} n_i p_{xi}$.
Meanwhile, in the case $S>1$, the event probability of that is given by
\begin{equation}
\label{eq:A.2}
\binom{N\!-\!1}{S\!-\!1} 
\cdot \!
\underbrace{
\sum_{\substack{s_1+\cdots+s_K\\=S-1,\\s_i\ge0}}
\! \binom{S\!-\!1}{s_1,\! \cdots \!,s_K} 
\! \prod^{K}_{j=1}(n_jp_{xj})^{s_j}
}_{\boxed{\text{a}}}
\cdot \!
\underbrace{
\sum_{\substack{t_1+\cdots+t_K\\=N-S,\\t_k\ge0}}
\! \binom{N\!-\!S}{t_1,\! \cdots \!,t_K} 
\! \prod^{K}_{l=1}(n_l(1-p_{xl}))^{t_l}
}_{\boxed{\text{b}}},
\end{equation}
where the terms $\boxed{\text{a}}$ and $\boxed{\text{b}}$ denote the probabilistic sum of combination of $K$-type strategies, respectively, among $S-1$ participants and among remaining $N-S$ nonparticipants.
With the combination of participants in Eq. (\ref{eq:A.2}), 
$(s_1, \cdots, s_K)$, the focal participant earns 
\begin{equation}
\label{eq:A.3}
\frac{r}{S} \left( \sum^{K}_{i=1} s_i c_{xi} + c_y \right) - c_y
= \frac{r}{S} \sum^{K}_{i=1} s_i c_{xi} + c_y \left( \frac{r}{S} -1\right).
\end{equation}
Considering the event probability, thus, the focal participant's expected payoff is
\begin{eqnarray}
\label{eq:A.4}
&& \binom{N-1}{S-1}
\boxed{\text{a}} \
\boxed{\text{b}}
\! \left[ \frac{r}{S} \! \sum^{K}_{k=1} s_k c_{xk} + c_y \left( \frac{r}{S}-1\right) \! \right]
\nonumber \\
&=&
\binom{N-1}{S-1}
\boxed{\text{b}}
\Biggl[
\boxed{\text{a}} \; c_y \left( \frac{r}{S}-1\right)
 +
\frac{r}{S} \sum^{K}_{k=1}
\underbrace{
\sum_{\substack{s_1+\cdots+s_K\\=S-1,\\s_i \ge 0}}
\! \binom{S-1}{s_1, \cdots ,s_K}
\! \left( \prod^{K}_{j=1}(n_jp_{xj})^{s_j} \right)
s_k 
}_{\boxed{\text{c}}}
c_{xk} \Biggr]. \nonumber \\
&&
\end{eqnarray}
From the multinomial series expansion,
the terms $\boxed{\text{a}}$ and $\boxed{\text{b}}$ are identical to $(n_1p_{x1}+\cdots+n_Kp_{xK})^{S-1}=\bar{p}^{S-1}_x$
and $(n_1(1-p_{x1})+\cdots+n_K(1-p_{xK}))^{N-S}=(1-\bar{p}_x)^{N-S}$, respectively.
Since
$\binom{S\!-\!1}{s_1,\cdots,s_k,\cdots,s_K} s_k = \binom{S\!-\!2}{s_1,\cdots,s_k\!-\!1,\cdots,s_K} (S\!-\!1)$
for $1\leq \forall{k} \leq K$,
the term $\boxed{\text{c}}$ in Eq. (\ref{eq:A.4}) can be rewritten as
\begin{eqnarray}
\label{eq:A.5}
&& 
\sum_{\substack{s_1+\cdots+s_K\\=S-1,\\s_i\ge0}}
\! \binom{S\!-\!1}{s_1,\! \cdots \!,s_k,\! \cdots \!,s_K}
\! \left( \prod^{K}_{j=1}(n_jp_{xj})^{s_j} \right)
s_k 
\nonumber \\
&=& 
(S\!-\!1) (n_kp_{xk})\!\!\!\!
\sum_{\substack{s_1+\cdots+s'_k+\\ \cdots+s_K=S-2,\\s_i\ge0, s'_k\ge0}}
\! \binom{S\!-\!2}{s_1,\! \cdots \!,s'_k,\! \cdots \!,s_K}
\! (n_1p_{x1})^{s_1}\cdots(n_kp_{xk})^{s'_k}\cdots(n_Kp_{xK})^{s_K}
\nonumber \\
&& \text{(where} \  s'_k := s_k -1 \ \text{)}
\nonumber \\
&=&
(S\!-\!1) (n_kp_{xk}) (n_1p_{x1}+\cdots+n_Kp_{xK})^{S-2}
\nonumber \\
&\!=\!&
(S\!-\!1) (n_kp_{xk})\bar{p}^{S-2}_x.
\end{eqnarray}
Thus, we obtain the analytical expression of $g(X,y)$, as follows:
\begin{eqnarray} \label{eq:A.6} 
g(X,y) &=&
(1-\bar{p}_x)^{N-1} \sigma +  \sum_{S=2}^{N} 
\binom{N-1}{S-1}
(1-\bar{p}_x)^{N-S}
\left[
\bar{p}^{S-1}_x c_y \left( \frac{r}{S}-1 \right)
+ \bar{p}^{S-2}_x \frac{r(S-1)}{S} \sum^{K}_{i=1} n_ip_{xi}c_{xi}
\right]
\nonumber \\
&=& 
(1-\bar{p}_x)^{N-1} \sigma +  \sum_{S=2}^{N} 
\binom{N-1}{S-1}
\bar{p}^{S-1}_x
(1-\bar{p}_x)^{N-S}
\left[
 c_y \left( \frac{r}{S}-1 \right)
+ \frac{r(S-1)}{S} \sum^{K}_{i=1} \frac{n_ip_{xi}}{\bar{p}_x} c_{xi}
\right]
\nonumber \\
&=& 
(1-\bar{p}_x)^{N-1} \sigma +  \sum_{S=2}^{N} 
\binom{N-1}{S-1}
\bar{p}^{S-1}_x
(1-\bar{p}_x)^{N-S}
\left[
\frac{r}{S} \left( \! (S-1)\sum^{K}_{i=1} \frac{n_ip_{xi}}{\bar{p}_x} c_{xi} \!+\! c_y \! \right) \!-\! c_y
\right].
\end{eqnarray}

We note that Eq. (\ref{eq:A.6}) can be applied to the case that the rate mutant $y=(c_y,p_y)$ has a small yet finite mass $\epsilon>0$, with considering the mutant's patch $(x_{K+1},n_{K+1})=(y,\epsilon)$ in the resident population of $K+1$ patches, $X=\{ (x_i, n_i) \}_{i=1,\cdots, K+1}$ with $K \ge 1$. 
It is clear that as $\epsilon \to 0$, we can recover Eq. (\ref{eq:A.6}) for the rare mutant with the infinitesimally small mass and the resident with $K$ patches; in particular Eq. (\ref{eq:2.3}) for $K=1$.  

\subsubsection*{Mapping polymorphism to monomorphism}
This yields a natural mapping from polymorphic populations to monomorphic populations, such that it keeps the expected payoff of a mutant unchanged.
Let us put
\begin{eqnarray}
\label{eq:A.7}
c_{x^*}&:=& \sum^{K}_{i=1} \frac{n_ip_{xi}}{\bar{p}_x} c_{xi}, \\
\label{eq:A.8}
p_{x^*}&:=& \bar{p}_x,
\end{eqnarray}
which define a mapping from the set of polymorphic populations with finite support to the set of monomorphic populations, as follows:
\begin{equation}
\label{eq:A.9}
\pi:X \longmapsto x^*:=(c_{x^*},p_{x^*}).
\end{equation}
Using this, we represent $g(X,y)$ in Eq. (\ref{eq:A.6}), as follows:
\begin{eqnarray}
\label{eq:A.10}
g(X,y) &=&
  (1- p_{x^*})^{N-1} \sigma + \sum_{S=2}^{N} \binom{N-1}{S-1} {p_{x^*}}^{S-1} (1-p_{x^*})^{N-S} 
 \left[ \frac{r ((S-1)c_{x^*} + c_y)}{S} - c_y \right]
 \nonumber \\
&=&
  (1- p_{x^*})^{N-1} \sigma + (1-(1-p_{x^*})^{N-1}) (r-1)c_{x^*} - (c_y-c_{x^*}) F(1-p_{x^*}) ,
\end{eqnarray}
where $F(z)$ is in Eq. (\ref{eq:2.4}).
In particular, if a resident population $X$ is monomorphic with strategy $x$, Eq. (\ref{eq:A.10}) is consistent with the former Eq. (\ref{eq:2.3}), that is, $g(X,y)=g(x^*,y)$ and thus $P(X,y)=P(x^*,y)$.
Therefore, in this continuous-strategy game, the expected payoff of a mutant within any polymorphic residents with finite support can be calculated  by using the mapping $\pi$ in Eq. (\ref{eq:A.9}).
Let us clarify the difference between the strategy $x^*=(c_{x^*},p_{x^*})$ and the average strategy $\bar{x}=(\bar{c}_x,\bar{p}_x)$ given by $(\sum^K_{i=1}n_i c_{xi}, \sum^K_{i=1}n_i p_{xi}$).
Obviously, $p_{x^*}=\bar{p}_x$, yet
\begin{equation}
\label{eq:A.12}
c_{x^*}-\bar{c}_x=\frac{\sum^K_{i=1} n_i n_j (c_{xi}-c_{xj})(p_{xi}-p_{xj})}{\sum^K_{i=1}n_i p_{xi}}
= \frac{Cov(c_x,p_x)}{\bar{p}_x}.
\end{equation}

\subsubsection*{Extend to polymorphic mutants}

Moreover, we extend the expected payoff to a mutant with probability distribution with finite support, as resident populations. 
Let us represent the mutant's distribution as $Y=\{ (y_i, m_i) \}_{1 \leq i \leq K'}$, where $y_i=(c_{yi},p_{yi})$ and $m_i$ is the relative frequency of strategy $y_i$ ($1 \leq i \leq K'$).
We define the expected payoff of a mutant with distribution $Y$ in a resident population with distribution $X$, by weighted sum of the expected payoff of $Y$'s each patch, as follows:    
\begin{eqnarray}
\label{eq:A.13}
P(X,Y):=\sum^{K'}_{i=1} m_i P(X,y_i)=\sum^{K'}_{i=1} m_i P(x^*,y_i).
\end{eqnarray}
From Eq. (\ref{eq:A.10}),
\begin{eqnarray}
\label{eq:A.14}
&& \sum^{K'}_{i=1} m_i p_{yi} g(X,y_i) \nonumber \\
&&= \left( {\sum^{K'}_{i=1} m_i p_{yi}} \right) \!\!
\left[(1\!-\! p_{x^*})^{N-1} \sigma + (1\!-\!(1\!-\!p_{x^*})^{N-1}) (r\!-\!1)c_{x^*} 
- \left( \frac{\sum^{K'}_{i=1}m_i p_{yi} c_{yi}}{\sum^{K'}_{i=1} m_i p_{yi}}-c_{x^*} \right) \! F(1\!-\!p_{x^*})  
\right] \nonumber \\
&& = p_{y^*} g(X,y^*),
\end{eqnarray}
where $y^*= (c_{y^*},p_{y^*})$ is given by $\pi(Y)=(\sum^{K'}_{i=1} m_i p_{yi}, 
\sum^{K'}_{i=1}m_i p_{yi} c_{yi} / \sum^{K'}_{i=1} m_i p_{yi} )$.
Thus, we obtain
\begin{eqnarray}
\label{eq:A.15}
\sum^{K'}_{i=1} m_i P(x^*,y_i)&=&\sum^{K'}_{i=1} m_i (p_{yi} g(X,y_i) + (1-p_{yi}) \sigma) \nonumber \\
&=& p_{y^*} g(X,y^*) + (1-p_{y^*}) \sigma \nonumber \\
&=& P(x^*,y^*), 
\end{eqnarray}
and then,
\begin{equation}
\label{eq:A.16}
P(X,Y)=P(\pi(X),\pi(Y)).
\end{equation}

Finally, we define invasion fitness in the case that the mutant and the resident population have strategy distributions with finite support, respectively $Y$ and $X$. Similarly, invasion fitness is given by
\begin{equation}
\label{eq:A.17}
S(X,Y) = P(X,Y)-P(X,X),
\end{equation}
Substituting Eqs. (\ref{eq:A.1}) and (\ref{eq:A.9}) into Eq. (\ref{eq:A.17}), we can have the Taylor expansion around $x^*=(c_{x^*},p_{x^*})$: 
\begin{eqnarray}
\label{eq:A.18}
S(X,Y)
&=& -p_{x^*} F(z_{x^*})(c_{y^*} - c_{x^*}) + ((r - 1)c_{x^*} - \sigma) (1 - (1 - p_{x^*})^{N-1}) (p_{y^*} - p_{x^*}) \nonumber \\
&&- F(z_{x^*})(c_{y^*} - c_{x^*})(p_{y^*} - p_{x^*}).
\end{eqnarray} 

\subsection*{A.2 \ Curvatures}

At a given focal point $x$, the curves $C_m$ and $C_r$ have identical slopes.
To determine the local arrangement of $C_m$ and $C_r$ around $x$, thus we need to calculate those second derivatives, that is, curvatures.
In the case of $C_m$, we have
\begin{eqnarray}
\label{eq:A.19}
p''_y(x) 
=\left. \frac{\partial^2 p_y}{\partial c^2_y} \right|_{y=x}
= 2p_x \left( \frac{a_{11}(x)}{a_{01}(x)} \right)^2,
\end{eqnarray}
and then the curvature of $C_m$ at $x$ is given by $\kappa_m(x) = p''_y(x) [1+ p'_y(x)]^{-3/2}$,
where $p'_y(x) = \partial p_y / \partial c_y |_{y=x} = -p_x a_{11}(x)/a_{01}(x)$. 
In the case of $C_r$, we have
\begin{eqnarray}
\label{eq:A.20}
p''_y(x) = 2p_x \frac{a_{11}(x)}{a_{01}^2(x)} \left( a_{01}'(x) + p_x F'(1-p_x) \right),
\end{eqnarray}
where $a_{01}'(x) = \left. \partial a_{01}(y) / \partial c_y \right|_{y=x} $. 
We can then calculate the curvature $\kappa_r(x)$ of $C_r$ at $x$, as well as $\kappa_m(x)$.

Let us investigate the signs of $\kappa_m(x)$, $\kappa_r(x)$, and $\kappa_r(x)-\kappa_m(x)$ given by
\begin{equation}
\label{eq:A.21}
\frac{2p_x}{ [1+ p'_y(x)]^{3/2} } \left( \frac{a_{11}(x)}{a_{01}(x)} \right)
\left( a_{01}'(x) + p_x F'(1-p_x) - a_{11}(x) \right)
\end{equation}
From Eq. (\ref{eq:A.19}), the sign of $\kappa_m(x)$ is always positive. As for the sign of $\kappa_r(x)$, we should have a further calculation in the Eq. (\ref{eq:A.20}), as follows: putting $z_x=1-p_x$,
\begin{eqnarray}
\label{eq:A.22}
&& a_{01}'(x) + (1-z_x) F'(z_x) \nonumber \\
&& = - (1 -z_x^{N-1}) + r \left(
1-\frac{1-z_x^N}{N(1-z_x)}
\right) \left(
1+\frac{(N-1)(1-z_x)z_x^{N-2}}{1-z_x^{N-1}}
\right) \nonumber \\
&& \geq \left[
- (1\!-\!z_x^{N-1}) \!+\! 2 \! \left( 1\!-\! \frac{1-z_x^N}{N(1\!-\!z_x)} \right) 
\right] 
\!+ r \! \left(
1\!-\! \frac{1\!-\!z_x^N}{N(1\!-\!z_x)}
\right) \!\! \left(
1\!+\! \frac{(N\!-\!1)(1\!-\!z_x)z_x^{N-2}}{1\!-\!z_x^{N-1}}
\right) \! , \nonumber \\
\end{eqnarray}  
where all the terms are positive if $z_x \neq 1$, and otherwise zero. 
We note that
\begin{eqnarray}
\label{eq:A.23}
&& - (1 -z_x^{N-1}) + 2 \left( 1-\frac{1-z_x^N}{N(1-z_x)} \right) \nonumber \\
&& = (1-z_x)^2 \left[
(N-2)(1+\cdots+z_x^{N-3}) + \cdots + (N-2k)(z_x^{k-1}+\cdots+z_x^{(N-2)-k}) \right. \nonumber \\
&& \qquad + \ \cdots +
\begin{cases}  
\left. 2 (z_x^{\frac{N}{2}-2} + z_x^{\frac{N}{2}}) \right]  \qquad \text{if $N$ is even,} \nonumber \\
\left. z^{\frac{N-3}{2}} \right]  \qquad\qquad\qquad \text{if $N$ is odd,} 
\end{cases} \nonumber \\
&& \geq 0
\end{eqnarray}
Thus, the sign of $\kappa_r(x)$ is equal to the sign of $a_{11}(x)$. 
As for the sign of $\kappa_r(x)-\kappa_m(x)$, from Eqs. (\ref{eq:A.19}) and (\ref{eq:A.20}), we need to know the sign of the term as follows:
\begin{eqnarray}
\label{eq:A.24}
&& a_{01}'(x) + (1-z_x) F'(z_x) - a_{11}(x) \nonumber \\
&& = r \left[
\left( \!
- (1 \!-\! z_x^{N-1}) \!+\! 2 \! \left( 1\! - \! \frac{1-z_x^N}{N(1\!-\!z_x)} \right)
\! \right) 
\!+\!
\frac{(N\!-\!1)(1\!-\!z_x)z_x^{N-2}}{1-z_x^{N-1}}
\left( 1\!-\! \frac{1-z_x^N}{N(1\!-\!z_x)} \right)
\right] \nonumber \\
&&
\end{eqnarray}
From Eq. (\ref{eq:A.23}), the above is positive if $z_x \neq 1$, and otherwise, zero.
Therefore, the sign of $\kappa_r(x)-\kappa_m(x)$ is equal to the sign of $a_{11}(x)$, as well as $\kappa_r(x)$.

\subsection*{A.3 \ Replicator dynamics for two strategies}

We analyze frequency dynamics between two strategies with $x_1=(c_1, p_1)$ and $x_2=(c_2, p_2)$. We assume that the growth rate of strategy is determined by the replicator dynamics.  
It has already been studied in a special case of $c_1=0$ and $c_2=1$, that is, consisting of full defection and full cooperation which are extended to be with probabilities to participate in the public good game \cite{SasEtal07}. The replicator dynamics in the special case have been classified into four fundamental types of evolutionary scenario in two-strategy games, given by dominance, coexistence, bi-stability, and neutrality.   
Here we present a general classification for two arbitrary strategies $x_1$ and $x_2$ in the strategy space $U$. It then turns out that similarly, the evolutionary scenario can be determined by the combination of signs, given by $(\text{sgn}\:S(x_2,x_1),\text{sgn}\:S(x_1,x_2))$. 

Let us denote  as $X_h$ a mixed state between $x_1$ and $x_2$, with relative frequencies $h$ and $h-1$ respectively.
To investigate the evolutionary fate of such the dimorphic population, we should know the payoff difference between  $x_1$ and $x_2$ in the environment set by the mixed state $X_h$, that is, the advantage function given by
\begin{equation}
\label{eq:A.25}
P(X_h,x_1)-P(X_h,x_2)=:\tilde{F}_{12}(h),
\end{equation}
where $P(X_h,x_i)$ denotes the expected payoff of strategy $x_i$ ($i=1,2$) within the mixed population $X_h$.
The average payoff over the population is given by $P(X_h,X_h) = h P(X_h,x_1) + (1-h) P(X_h,x_2)$.
Using the relative average payoff $S(X_h,x_i):=P(X_h,x_i)-P(X_h,X_h)$ ($i=1,2$), we obtain 
\begin{equation}
\label{eq:A.26}
\tilde{F}_{12}(h)=S(X_h,x_1)-S(X_h,x_2), 
\end{equation}
in particular,
\begin{equation}
\label{eq:A.27}
\tilde{F}_{12}(0)=S(x_2,x_1), \quad \tilde{F}_{12}(1)=-S(x_1,x_2).
\end{equation}
The replicator equation is given by
\begin{eqnarray}
\label{eq:A.28}
\dot{h} &=& h (P(X_h,x_1)-P(X_h,X_h)) = h S(X_h,x_1) \nonumber \\
&=& h(1-h) (P(X_h,x_1)-P(X_h,x_2)) \nonumber \\
&=& h(1-h) \tilde{F}_{12}(h).
\end{eqnarray}  
From a straightforward calculation, $\tilde{F}_{12}(h)$ can be rewritten as
\begin{eqnarray}
\label{eq:A.29}
\tilde{F}_{12}(h) &=&
\underbrace{ [ (c_1-c_{\textrm{Q}})p_1 - (c_2-c_{\textrm{Q}})p_2 ] }_{\text{(A)}}
\underbrace{(r-1)(1-z^{N-1}_h)}_{\text{(Z1)}}
 + \underbrace{(c_2-c_1)p_1p_2}_{\text{(B)}}
\underbrace{
\frac{r}{1-z_h} \left(
1-\frac{1-z^N_h}{N(1-z_h)}
\right)
}_{\text{(Z2)}},  \nonumber \\
&&
\end{eqnarray}
 where $z_h = 1-p_h = 1-(hp_1+(1-h)p_2)$ and $c_{\textrm{Q}}=\frac{\sigma}{r-1}$. 
In particular, if $p_1=p_2 \ (=:p)$, Eq. (\ref{eq:A.29}) is reduced, as follows:
\begin{equation}
\label{eq:A.31}
\tilde{F}_{12}(h) =
(c_2-c_1)p  \left[
-(r-1)(1-z^{N-1})+r\left( 1-\frac{1-z^N}{N(1-z)} \right)
\right] = (c_2-c_1)p \cdot F(z),
\end{equation}
where $z=1-p$ which is independent of $h$.
In the reduced case, thus $\tilde{F}_{12}(h)$ must constant and the replicator dynamics is unilateral toward either of the two extreme states.
Which one is a global attractor depends on two kinds of magnitude relation: $c_1$ and $c_2$; and, if $r>2$, also $p$ and $p_{\textrm{Q}}$, where $p_{\textrm{Q}}$ is  the unique interior root of $F(1-p_{\textrm{Q}})$. If $p=0$, the dynamics is neutral.   
In the following section we consider the general case: $p_1 \neq p_2$.

We first consider this general case that \textit{both (A) and (B) are non-zero}.
 If the boundary values $\tilde{F}_{12}(0)$ and $\tilde{F}_{12}(1)$ have signs that are opposite to each other 
($\iff$ $S(x_2, x_1)$ and $S(x_1,x_2)$ have same signs, from Eq. (\ref{eq:A.27})), 
$\tilde{F}_{12}(h)$ has to have at least one interior root within $[0,1]$. In the case we show that $\tilde{F}_{12}(h)$ is monotonic and thus the interior root is unique, as follows.
Since there exists $h$ such that $\text{Eq. (\ref{eq:A.28})} =0$ holds, the term (A) $<0$, if the term (B) $>0$ ($\Leftrightarrow c_2>c_1$), or the term (A) $>0$, if the term (B) $<0$ ($\Leftrightarrow c_2<c_1$).
Differentiate $\tilde{F}_{12}(h)$ with respect to $h$,
\begin{eqnarray}
\label{eq:A.32}
\frac{d \tilde{F}_{12}}{d h} 
= \frac{d \tilde{F}_{12}}{dz_h} \frac{d z_h}{d h}
&=& \left[
\underbrace{ ( (c_1-c_{\textrm{Q}})p_1 - (c_2-c_{\textrm{Q}})p_2 ) }_{\text{(A)}} 
\underbrace{ \cdot -(r-1)(N-1)z^{N-2}_h }_{\text{(Z1')}} 
\right. \nonumber \\
&& +  \left.
\underbrace{ (c_2-c_1)p_1p_2 }_{\text{(B)}} 
\underbrace{
\frac{r}{N} ( (N\!-\!2)\!+\!2(N\!-\!3)z_h \!+\! \cdots \!+\! (N\!-\!2)z^{N-3} )
}_{\text{(Z2')}} 
\right] \! (p_2-p_1) . \nonumber \\
\
\end{eqnarray}
Since the term (Z1') $<0$ and the term (Z2') $>0$ for all $t \in (0,1)$,
\begin{eqnarray}
\label{eq:A.34}
\frac{d \tilde{F}_{12}}{d h}  
  \begin{cases}
  >0 \quad \text{ if $(c_1-c_2)(p_1-p_2)>0$ }, \\
  <0 \quad \text{ if $(c_1-c_2)(p_1-p_2)<0$ }.
  \end{cases}
\end{eqnarray}
Therefore, $\tilde{F}_{12}(h)$ is monotonic and the interior root is unique if it exists.
If $(\text{sgn}\:S(x_2,x_1),\text{sgn}\:S(x_1,x_2))$ is $(-,-)$, we have that $\tilde{F}_{12}(0)<0$ and $\tilde{F}_{12}(1)>0$ in Eq. (\ref{eq:A.27})), and the monotonicity leads to that $\tilde{F}_{12}(h)$ is increasing. Thus the interior fixed point exists and is a repellor (``bi-stability'').
If $(\text{sgn}\:S(x_2,x_1),\text{sgn}\:S(x_1,x_2))$ is $(+,+)$, similarly, $\tilde{F}_{12}(0)>0$, $\tilde{F}_{12}(1)<0$, and $\tilde{F}_{12}(h)$ is decreasing. Thus the interior fixed point exists and is an attractor (``coexistence'').
If $(\text{sgn}\:S(x_2,x_1),\text{sgn}\:S(x_1,x_2))$ is $(+,-)$, we have that both $\tilde{F}_{12}$ and $\tilde{F}_{12}(1)>0$.
Thus, $\tilde{F}_{12}(h)$ has no interior root and the uniform state with $x_1$ is a global attractor (``$x_1$-dominance'').
Finally, if $(\text{sgn}\:S(x_2,x_1),\text{sgn}\:S(x_1,x_2))$ is $(-,+)$, then that both $\tilde{F}_{12}$ and $\tilde{F}_{12}(1)<0$.
Similarly, thus there exists no interior root and the other uniform state with $x_s$ is a global attractor (``$x_s$-dominance'').

We turn to the boundary cases. 
Let us first assume that \textit{only (B) equals zero}. 
It follows that (B) $=0$ $\iff$ $c_1=c_2$, $p_1=0$ or $p_2=0$.  
In the case of $c_1=c_2 \ (=:c)$, $\tilde{F}_{12}(h)$ is reduced to
\begin{equation}
\label{eq:A.35}
\tilde{F}_{12}(h) =
(c-c_Q)(p_1-p_2)(r-1)(1-z^{N-1}_h).
\end{equation} 
The sign of $\tilde{F}_{12}(h)$ above is unchanged, and thus the direction of dynamics is unilateral.
A global attractor is determined by the sign of $(c-c_{\textrm{Q}})(p_1-p_2)$.
Similarly, in the case of $p_2=0$, if $c_1 > c_{\textrm{Q}}$, the dynamics is $x_1$-dominance, or otherwise, $x_2$-dominance;
in the case of $p_1=0$, if $c_2 > c_{\textrm{Q}}$, the dynamics is $x_2$-dominance, or otherwise, $x_1$-dominance.
If \textit{only (A) equals zero}, it follows from Eq. (\ref {eq:A.32}) that the dynamics is $x_1$-dominance if $c_2>c_1$, or otherwise, $x_2$-dominance.
Finally, if \textit{both (A) and (B) equal zero}, the dynamics is neutral.

\subsection*{A.4 \ Continuous voluntary public good games with stand-alone play}
We examine continuous public good game with stand-alone play.
We assume that a single participant with contribution level $c$ would receive $rc-c$ from the good provision, instead of the loner's payoff $\sigma$. 
In this case it is clear that if the contribution level is greater than $c_\textrm{Q}=\sigma/(r-1)$, nonparticipation is no longer individually rational: each point of the line $p=0$ is no longer a Nash equilibrium.
As in Sect. 2.2, the probability that a mutant player with strategy $y$ finds itself among the $S -1$ resident co-players with strategy $x$ is $\binom{N - 1}{S - 1} p_x^{S-1} (1-p_x)^{N-S}$, yet the number of players $S$ can vary between 1 and $N$. Thus, Eq. (\ref{eq:2.3}) turns into
\begin{eqnarray}
\label{eq:A.39}
g(x,y) 
&=& \sum_{S=1}^{N}{\binom{N\!-\!1}{S\!-\!1} p_x^{S-1} (1-p_x)^{N-S} \left[ \frac{r((S-1)c_x+c_y)}{S}-c_y \right] } \nonumber \\
&=&   (r-1)c_x - (c_y - c_x) F_0(1-p_x), 
\end{eqnarray}
where
$F_0 (z) = 1-r(1-z^N)/(N(1-z))$. 
We note that for $1<r<N$, $F_0 (z)$ is monotonically decreasing and has a unique root in the open interval $(0,1)$.
Hence, using $c_{\textrm{Q}}=\sigma/(r-1)$,
\begin{equation}
\label{eq:A.41}
P(x,y) = \sigma + (r-1)(c_x - c_{\textrm{Q}}) p_y - (c_y - c_x) p_y F_0 (1-p_x).
\end{equation}
Indeed, this yields a similar system, as follows:
$\dot{c} = -p F_0(1-p)$ and $\dot{p} = (r-1)(c-c_\textrm{Q})$.
in contrast to the original system, the variant system is so simple that has only two singular points, a center point and a boundary saddle point with $p=0$. This center point exists for all $r$ within $1<r<N$.


\begin{thebibliography}{}
%
%

\bibitem{Tri71}
Trivers RL (1971). The evolution of reciprocal altruism. Q Rev Biol 46:35--57.
doi:10.1086/406755
\bibitem{WiSo94}
Wilson DS, Sober E (1994). Reintroducing group selection to the human behavioral sciences. Behav Brain Sci 17:585--654.
doi:10.1017/s0140525x00036104
\bibitem{Har68}
Hardin G (1968) The tragedy of the commons. Science 162:1243--1248. 
doi:10.1126/science.162.3859.1243
\bibitem{Ost90}
Ostrom E (1990) Governing the Commons: the Evolution of Institutions for Collective Action. Cambridge University Press, New York

\bibitem{Dawes}
Dawes RM (1980) Social dilemmas. Annu Rev Psychol 31:169--193.
doi:10.1146/annurev.ps.31.020180.001125

\bibitem{Bin94}
Binmore KG (1994) Playing Fair: Game Theory and the Social Contract. MIT Press, Cambridge, MA
\bibitem{Cha11} 
Chaudhuri A (2011) Sustaining cooperation in laboratory public goods experiments: a selective survey of the literature. Exp Econ 14:47--83.
doi:10.1007/s10683-010-9257-1

\bibitem{Sig10} 
Sigmund K (2010) The Calculus of Selfishness. Princeton University Press, Princeton, NJ.

\bibitem{Hir70}
Hirschman AO (1970) Exit, voice, and loyalty. Harvard University Press, Cambridge
\bibitem{OrbDaw93}
Orbell JM, Dawes RM (1993) Social welfare, cooperators' advantage, and the option of not playing the game.  Am Soc Rev 58:787--800.
doi:10.2307/2095951
\bibitem{BatKit95}
Batali J, Kitcher P (1995) Evolution of altruism in optional and compulsory games. J Theor Biol 175:161--171.
doi:10.1006/jtbi.1995.0128
\bibitem{Hayashi}
Hayashi N, Yamagishi T (1998) Selective play: choosing partners in an uncertain world. Pers Soc Psychol Rev 2:276--289.
doi: 10.1207/s15327957pspr0204\_4

\bibitem{HauEtal02a}
Hauert C, De Monte S, Hofbauer J, Sigmund K (2002) Volunteering as Red Queen mechanism for cooperation in public goods games. Science 296:1129--1132.
doi:10.1126/science.1070582
\bibitem{HauEtal02b}
Hauert C, De Monte S, Hofbauer J, Sigmund K (2002) Replicator dynamics for optional public good games. J Theor Biol 218:187--194.
doi:10.1006/jtbi.2002.3067


\bibitem{wahl1}
Wahl LM, Nowak MA (1999) The continuous prisoner's dilemma: I. Linear reactive strategies. 
J Theor Biol 200:307--321.
doi:10.1006/jtbi.1999.0996
\bibitem{wahl2}
Wahl LM, Nowak MA (1999) The continuous prisoner's dilemma: II. Linear reactive strategies with noise. J Theor Biol 200:323--338.
doi:10.1006/jtbi.1999.0997
\bibitem{killingback3}
Killingback T, Doebeli M, Knowlton N (1999) Variable investment, the continuous prisoner's dilemma, and the origin of cooperation. Proc R Soc B 266:1723--1728.
10.1098/rspb.1999.0838
\bibitem{killingback1}
Killingback T, Doebeli M (2002) The continuous prisoner's dilemma and the evolution of cooperation through reciprocal altruism with variable investment. Am Nat 160:421--438.
doi:10.1086/342070
\bibitem{sherratt}
Sherratt TN, Roberts G (2002) The stability of cooperation involving variable investment. J Theor Biol 215:47--56.
10.1006/jtbi.2001.2495
\bibitem{Doe04}
Doebeli M, Hauert C, Killingback T (2004) The evolutionary origin of cooperators and defectors. Science 306:859--862.
doi:10.1126/science.1101456
\bibitem{Bro08}
Brown JS, Vincent TL (2008) Evolution of cooperation with shared costs and benefits. Proc R Soc B 275:1985--94.
doi:10.1098/rspb.2007.1685
\bibitem{Ake}
Br{\"a}nnstr{\"o}m {\AA}ke, Gross T, Blasius B, Dieckmann U (2011) Consequence of fluctuating group size for evolution of cooperation. J Math Biol 63:263--281.
doi:10.1007/s00285-010-0367-3
\bibitem{CreEtal11}
Cressman R, Song JW, Zhang BY, Tao Y (2011) Cooperation and evolutionary dynamics in the public goods game with institutional incentives. J Theor Biol 299:144--151. doi:10.1016/j.jtbi.2011.07.030
\bibitem{Deng11}
Deng K, Chu T (2011) Adaptive evolution of cooperation through Darwinian dynamics in public goods games. PLoS ONE 6:e25496.
doi:10.1371/journal.pone.0025496
\bibitem{Zhang}
Zhang Y, Fu F, Wu T, Xie G, Wang L (2013) A tale of two contribution mechanisms for nonlinear public goods. Sci Rep 3:2021.
doi:10.1038/srep02021
\bibitem{Par13}
Parvinen K (2013) Joint evolution of altruistic cooperation and dispersal in a metapopulation of small local populations. Theor Popul Biol 85:12--19.
10.1016/j.tpb.2013.01.003

\bibitem{SasEtal07}
Sasaki T, Okada I, Unemi T (2007) Probabilistic participation in public goods games. Proc R Soc B 274:2639--2642.
doi:10.1098/rspb.2007.0673
\bibitem{Chen08}
Chen X, Fu F, Wang L (2008) Interaction stochasticity supports cooperation in spatial prisoner's dilemma. Phys Rev E 78:051120.
doi:10.1103/PhysRevE.78.051120

\bibitem{Oech02}
Oechssler J, Riedel F (2002) On the dynamic foundation of evolutionary stability in continuous models. J Econ Theory 107:223--252. doi:10.1006/jeth.2001.2950
\bibitem{dieckmann}
Dieckmann U, Law R (1996) The dynamical theory of coevolution: a derivation from stochastic ecological processes. J Math Biol 34:579--612.
doi:10.1007/BF02409751
\bibitem{geritz}
Geritz SAH, Kisdi \'{E}, Mesz\'{e}na G, Metz JAJ (1998) Evolutionarily singular strategies and the adaptive growth and branching of the evolutionary tree. Evol Ecol 12:35--57.
doi:10.1023/A:1006554906681
\bibitem{HofSig98}
Hofbauer J, Sigmund K (1998) Evolutionary Games and Population Dynamics. Cambridge University Press, Cambridge, UK
\bibitem{Ake13}
Br\"{a}nnstr\"{o}m \r{A}, Johansson J, von Festenberg N (2013) The hitchhiker's guide to adaptive dynamics. Games 4(3):304--328. doi:10.3390/g4030304
\bibitem{Meszena}
Mesz\'{e}na G, Kisdi \'{E}, Dieckmann U, Geritz SAH, Metz JAJ (2001) Evolutionary optimisation models and matrix games in the unified perspective of adaptive dynamics. Selection 2:193--210.
doi:10.1556/Select.2.2001.1-2.14
\bibitem{Lei09}
Leimar O (2009) Multidimensional convergence stability. Evol Ecol Res 11:191--208


\bibitem{hirsch} 
Hirsch MW, Smale S, Devaney RL (2004) Differential Equations, Dynamical Systems \& An Introduction to Chaos. 2nd ed. Elsevier, San Diego


\bibitem{may73}
Maynard Smith J, Price GR (1973) The logic of animal conflict. Nature 246:15--18. doi:10.1038/246015a0
\bibitem{Less90}
Lessard S (1990) Evolutionary stability: one concept, several meanings. Theor Popul Biol 37:159--170.
doi:10.1016/0040-5809(90)90033-R

\bibitem{Bom90}
Bomze IM (1990) Dynamical aspects of evolutionary stability. Mon Hefte Math 110:189--206.
doi:10.1007/BF01301675
\bibitem{Esh83}
Eshel I (1983) Evolutionary and continuous stability. J Theor Biol 103:99--111.
doi:10.1016/0022-5193(83)90201-1
\bibitem{Apo97}
Apaloo J (1997) Revisiting strategic models of evolution: The concept of neighborhood invader strategies. Theor Popul Biol 52:52--71.
doi:10.1006/tpbi.1997.1318


\bibitem{Hof09}
Hofbauer J, Oechssler J, Riedel F (2009) Brown--von Neumann--Nash dynamics: The continuous strategy case. Games Econ Behav 65:406--429.
doi:10.1016/j.geb.2008.03.006

\bibitem{geritz2}
Geritz SAH, van der Meijden E, Metz JAJ (1999) Evolutionary dynamics of seed size and seedling competitive ability. 
Theor Popul Biol 55:324--343.
doi:10.1006/tpbi.1998.1409
\bibitem{kisdi}
Kisdi \'{E}, Jacobs FJA, Geritz SAH (2001) Red Queen evolution by cycles of evolutionary branching and extinction. Selection 2:161--178.
doi:10.1556/Select.2.2001.1-2.12
\bibitem{vukics}
Vukics A, Asb\'{o}th J, Mesz\'{e}na G (2003) Speciation in multidimensional evolutionary space. Phys Rev E 68:041903.
doi:10.1103/PhysRevE.68.041903 
\bibitem{Nishimura}
Nishimura K, Isoda Y (2004) Variant evolutionary trees under phenotypic variance. J Theor Biol 226:79--87.
doi:10.1016/j.jtbi.2003.08.006
\bibitem{ito}
Ito HC, Dieckmann U (2007) A new mechanism for recurrent adaptive radiations. Am Nat 170:E96--E111.
doi:10.1086/521229
\bibitem{ito2}
Ito HC, Dieckmann U (2013) Evolutionary branching under slow directional evolution. J Theor Biol (Published online 5 Sep 2013). 
doi: 10.1016/j.jtbi.2013.08.028 

\bibitem{Cress05}
Cressman R (2005) Stability of the replicator equation with continuous strategy space. Math Soc Sci 50:127--147.
doi:10.1016/j.mathsocsci.2005.03.001
\bibitem{Cress06}
Cressman R, Hofbauer J, Riedel F (2006) Stability of the replicator equation for a single species with a multi-dimensional continuous trait space. J Theor Biol 239:273--288.
doi:10.1016/j.jtbi.2005.07.022

\bibitem{mazancourt}
de Mazancourt C, Dieckmann U (2004) Trade-off geometries and frequency-dependent selection. Am Nat 164:765--778.
doi:10.1086/424762


\bibitem{SemEtal03}
Semmann D, Krambeck HJ, Milinski M (2003) Volunteering leads to rock-paper-scissors dynamics in a public goods game. Nature 425:390--393.
doi:10.1038/nature01986
\bibitem{Akt04}
Aktipis CA (2004) Know when to walk away: contingent movement and the evolution of cooperation. J Theor Biol 231:249--260.
doi:10.1016/j.jtbi.2004.06.020
\bibitem{MatBoy09}
Mathew S, Boyd R (2009) When does optional participation allow the evolution of cooperation. Proc R Soc Lond B 276:1167--1174.
doi:10.1098/rspb.2008.1623
\bibitem{seki}
Sekiguchi T, Nakamaru M (2009) Effect of the presence of empty sites on the evolution of cooperation by costly punishment in spatial games. J Theor Biol 256:297--304.
doi:10.1016/j.jtbi.2008.09.025
\bibitem{IzqEtal10}
Izquierdo SS, Izquierdo LR, Vega-Redondo F (2010) The option to leave: conditional dissociation in the evolution of cooperation. J Theor Biol 267:76--84.
doi:10.1016/j.jtbi.2010.07.039
\bibitem{CasTor10}
Castro L, Toro MA (2010) Iterated prisoner's dilemma in an asocial world dominated by loners, not by defectors. Theor Popul Biol 74:1--5.
doi:10.1016/j.tpb.2008.04.001
\bibitem{XuEtal10}
Xu ZJ, Wang Z, Zhang LZ (2010) Bounded rationality in volunteering public goods games. J Theor Biol 264:19--23.
doi:10.1016/j.jtbi.2010.01.025
\bibitem{ZhongEtal13}
Zhong LX, Xu WJ, Shi YD, Qiu T (2013) Coupled dynamics of mobility and pattern formation in optional public goods games. Chaos Solitons Fractals 47:18--26. 
doi:10.1016/j.chaos.2012.11.012
\bibitem{Wu13}
Wu T, Fu F, Zhang Y, Wang L (2013) The increased risk of joint venture promotes social cooperation. PLoS ONE 8:e63801


\bibitem{Fow05}
Fowler J (2005) Altruistic punishment and the origin of cooperation. Proc Natl Acad Sci USA 102:7047--7049.
doi:10.1073/pnas.0500938102
\bibitem{BraEtal06}
Brandt H, Hauert C, Sigmund K (2006) Punishing and abstaining for public goods. Proc Natl Acad Sci USA 103:495--497.
doi:10.1073/pnas.0507229103
\bibitem{HauEtal07}
Hauert C, Traulsen A, Brandt H, Nowak MA, Sigmund K (2007) Via freedom to coercion: the emergence of costly punishment. Science 316:1905--1907.
doi:10.1126/science.1141588
\bibitem{SigEtal10}
Sigmund K, De Silva H, Traulsen A, Hauert C (2010) Social learning promotes institutions for governing the commons. Nature 466:861--863.
doi:10.1038/nature09203
\bibitem{SigEtal11}
Sigmund K, Hauert C, Traulsen A, De Silva H (2011) Social control and the social contract: the emergence of sanctioning systems for collective action. Dyn Games Appl 1:149--171.
doi:10.1007/s13235-010-0001-4
\bibitem{SasEtal12}
Sasaki T, Br\"{a}nnstr\"{o}m \AA, Dieckmann U, Sigmund K (2012) The take-it-or-leave-it option allows small penalties to overcome social dilemmas. Proc Natl Acad Sci USA 109:1165--1169.
doi:10.1073/pnas.1115219109
\bibitem{GarTra12}
Garc\'{i}a J, Traulsen A (2012) Leaving the loners alone: evolution of cooperation in the presence of antisocial punishment. J Theor Biol 307:168--173. 
doi:10.1016/j.jtbi.2012.05.011
\bibitem{Sa13}
Sasaki T (2013) The evolution of cooperation through institutional sanctioning and optional participation. 
Dyn Games Appl (Published online: 17 Aug 2013). 
doi:10.1007/s13235-013-0094-7


\bibitem{RanNow11}
Rand DG, Nowak MA (2011) The evolution of antisocial punishment in optional public goods games. Nat Commun 2:434.
doi:10.1038/ncomms1442


\bibitem{NaDi09}
Nakamaru M, Dieckmann U (2009) Runaway selection for cooperation and strict-and-severe punishment. J Theor Biol 257:1--8.
doi:10.1016/j.jtbi.2008.09.004
\bibitem{ShiNa13}
Shimao H, Nakamaru M (2013) Strict or graduated punishment? Effect of punishment strictness on the evolution of cooperation in continuous public goods games. PLoS ONE 8:e59894. doi:10.1371/journal.pone.0059894

\bibitem{Tra09}
Traulsen A, Hauert C, De Silva H, Nowak MA, Sigmund K (2009) Exploration dynamics in evolutionary games. Proc Natl Acad Sci USA 106:709--712.
doi:10.1073/pnas.0808450106

\bibitem{BiCa}
Bishop DT, Cannings C (1978) A generalized war of attrition. J Theor Biol 70:85--124. doi:10.1016/0022-5193(78)90304-1


\end{thebibliography}


\end{document}